\documentclass[10pt]{article}
\usepackage[numbers]{natbib}
\usepackage{marsden_article}
\usepackage{epstopdf}
\usepackage{setspace}
\usepackage[all,cmtip]{xy}

\usepackage{setspace,MnSymbol}
\usepackage{bbm,youngtab}
\usepackage{amscd,mathrsfs}
\usepackage{amsmath,dsfont}
\usepackage{euscript}
\usepackage{undertilde}

\usepackage[cbgreek]{textgreek} %For non-italic greek letters 

\usepackage{titlesec}
\titleformat{\subsubsection}[runin]
{\normalfont\bfseries}{\thesubsubsection}{1em}{}
\usepackage{caption3}

\DeclareGraphicsExtensions{.eps,.ps,.pdf}

\usepackage{hyperref}
\hypersetup{colorlinks=true, linkcolor=blue}
\hypersetup{colorlinks=true,citecolor=blue}

\newcommand{\comments}[1]{}
\usepackage{booktabs}
\usepackage{multirow}
\usepackage{enumerate}
 \theoremstyle{plain}	
 \newtheorem{thm}{Theorem}

 \theoremstyle{definition}	
 
 \newtheorem{rem}[thm]{Remark}
  
 \numberwithin{equation}{section}
\title{\Large{\textbf{Differential Complexes in Continuum Mechanics\footnote{To appear in \textit{Archive for Rational Mechanics and Analysis}.}}}}

\author{Arzhang Angoshtari\thanks{School of Civil and Environmental Engineering,
  Georgia Institute of Technology, Atlanta, GA 30332, USA. E-mail:  arzhang@gatech.edu.}
  \and Arash Yavari\thanks{School of Civil and Environmental Engineering \& The George W. Woodruff School of Mechanical Engineering,
  Georgia Institute of Technology, Atlanta, GA 30332, USA. E-mail:  arash.yavari@ce.gatech.edu.} }

\begin{document}

\maketitle

\begin{abstract}
We study some differential complexes in continuum mechanics that involve both symmetric and non-symmetric second-order tensors. In particular, we show that the tensorial analogue of the standard grad-curl-div complex can simultaneously describe the kinematics and the kinetics of motions of a continuum. The relation between this complex and the de Rham complex allows one to readily derive the necessary and sufficient conditions for the compatibility of the displacement gradient and the existence of stress functions on non-contractible bodies. We also derive the local compatibility equations in terms of the Green deformation tensor for motions of $2$D and $3$D bodies, and shells in curved ambient spaces with constant curvatures.  
\end{abstract}

\tableofcontents

%-----------------------------
%-----------------------------
%-----------------------------
\section{Introduction}

Differential complexes can provide valuable information for solving PDEs. The celebrated de Rham complex is a classical example. Let $\mathcal{B}$ be a $3$-manifold and let $\Omega^{k}(\mathcal{B})$ be the space of smooth\footnote{Throughout this paper, smooth means $C^{\infty}$.} $k$-forms on $\mathcal{B}$, i.e. $\boldsymbol{\alpha}\in\Omega^{k}(\mathcal{B})$ is an anti-symmetric $({}^{0}_{k})$-tensor with smooth components $\alpha_{i_{1}\cdots i_{k}}$. The exterior derivatives $d_{k}:\Omega^{k}(\mathcal{B})\rightarrow \Omega^{k+1}(\mathcal{B})$ are linear differential operators satisfying $d_{k+1}\circ d_{k}=0$, where $\circ$ denotes the composition of mappings.\footnote{When there is no danger of confusion, the subscript $k$ in $d_{k}$ is dropped.} Using the algebraic language, one can simply write the complex     
%-----------------------------
\begin{equation}
\scalebox{1}{\xymatrix@C=3ex{
0\ar[r] &\Omega^{0}(\mathcal{B}) \ar[r]^-{d} &\Omega^{1}(\mathcal{B}) \ar[r]^-{d} & \Omega^{2}(\mathcal{B}) \ar[r]^-{d} &\Omega^{3}(\mathcal{B}) \ar[r]  &0,  }} \nonumber
\end{equation}
%----------------------------- 
to indicate that $d$ is linear and the composition of any two successive operators vanishes. Note that the first operator on the left sends $0$ to the zero function and the last operator on the right sends $\Omega^{3}(\mathcal{B})$ to zero. The  above complex is called the de Rham complex on $\mathcal{B}$ and is denoted by $\left(\Omega(\mathcal{B}), d\right)$.

The complex property $d\circ d=0$, implies that $\mathrm{im}\,d_{k}$ (the image of $d_{k}$) is a subset of $\ker d_{k+1}$ (the kernel of $d_{k+1}$). The complex $\left(\Omega(\mathcal{B}), d\right)$ is exact if $\mathrm{im}\, d_{k}=\ker d_{k+1}$. Given $\boldsymbol{\beta}\in\Omega^{k}(\mathcal{B})$, consider the PDE $d\boldsymbol{\alpha}=\boldsymbol{\beta}$. Clearly, $\boldsymbol{\beta}\in\mathrm{im}\,d$ is the necessary and sufficient condition for the existence of a solution. If $\left(\Omega(\mathcal{B}), d\right)$ is exact, then $d\boldsymbol{\beta}=0$ guarantees that $\boldsymbol{\beta}\in\mathrm{im}\,d$. In general, the de Rham cohomology groups $\mathrm{H}^{k}_{dR}(\mathcal{B})=\ker d_{k}/\mathrm{im}\,d_{k-1}$ quantify the deviation of $\left(\Omega(\mathcal{B}), d\right)$ from being exact, i.e. this complex is exact if and only if all $\mathrm{H}^{k}_{dR}(\mathcal{B})$ are the trivial group $\{0\}$. 

If $\mathrm{H}^{k}_{dR}(\mathcal{B})$ is finite dimensional, then the celebrated de Rham theorem tells us that $\dim \mathrm{H}^{k}_{dR}(\mathcal{B})= b_{k}(\mathcal{B})$, where the $k$-th Betti number $b_{k}(\mathcal{B})$ is a purely topological property of $\mathcal{B}$. For example, if $\mathcal{B}$ is contractible, i.e. it does not have any holes in any dimension, then $b_{k}(\mathcal{B})=0$, $k\geq1$, or if $\mathcal{B}$ is simply-connected, then $b_{1}(\mathcal{B})=0$. On contractible bodies, $d\boldsymbol{\beta}=0$ is the necessary and sufficient condition for the solvability of $d\boldsymbol{\alpha}=\boldsymbol{\beta}$. If $\mathcal{B}$ is non-contractible, then the de Rham theorem \citep[Theorem 18.14]{Lee2012} tells us that $\boldsymbol{\beta}\in\mathrm{im}\,d$ if and only if
%-----------------------------
\begin{equation}\label{ExactForm}
d\boldsymbol{\beta}=0, \text{ and } \int_{c_{k}}\boldsymbol{\beta}=0, 
\end{equation}
%----------------------------- 
where for the purposes of this work, $c_{k}$ can be considered as an arbitrary closed (i.e. compact without boundary) $k$-dimensional $C^{0}$-manifold inside $\mathcal{B}$.\footnote{In fact, $c_{k}$ is a singular $k$-chain in $\mathcal{B}$ that can be identified with (a formal sum of) closed $k$-manifolds for integration, see standard texts such as \citep{Bredon1993,Lee2012} for the precise definition of $c_{k}$.} For more details on the de Rham complex, we refer the readers to the standard texts in differential geometry such as \citep{BottTu2010,Lee2012}. 

To summarize, we observe that the de Rham complex together with the de Rham theorem provide the required conditions for the solvability of $d\boldsymbol{\alpha}=\boldsymbol{\beta}$. Suppose that $\mathcal{B}$ is the interior of another manifold $\bar{\mathcal{B}}$. One can restrict $\left(\Omega(\mathcal{B}), d\right)$ to $\left(\Omega(\bar{\mathcal{B}}), d\right)$, where $\Omega^{k}(\bar{\mathcal{B}})$ is the space of those smooth forms in $\Omega^{k}(\mathcal{B})$ that can be continuously extended to the boundary $\partial\bar{\mathcal{B}}$ of $\bar{\mathcal{B}}$. If $\bar{\mathcal{B}}$ is compact, then $\left(\Omega(\bar{\mathcal{B}}), d\right)$ induces various Hodge-type decompositions on $\Omega^{k}(\bar{\mathcal{B}})$. Such decompositions allow one to study the above PDE subject to certain boundary conditions, e.g. see \citet{Schwarz1995,Gilkey1984}. On the other hand, it has been observed that differential complexes can be useful for obtaining stable numerical schemes. By properly discretizing the de Rham complex, one can develop stable mixed formulations for the Hodge-Laplacian \citep{Arnold2006,Arnold2010}.

Generalizing the above results for an arbitrary differential complex can be a difficult task, in general. This can be significantly simplified if one can establish a connection between a given complex and the de Rham complex. The grad-curl-div complex of vector analysis is a standard example. Let $C^{\infty}(\mathcal{B})$ and $\mathfrak{X}(\mathcal{B})$ be the spaces of smooth real-valued functions and smooth vector fields on $\mathcal{B}$. From elementary calculus, we know that for an open subset $\mathcal{B}\subset\mathbb{R}^{3}$, one can define the gradient operator $\mathrm{grad}:C^{\infty}(\mathcal{B})\rightarrow\mathfrak{X}(\mathcal{B})$, the curl operator $\mathrm{curl}:\mathfrak{X}(\mathcal{B})\rightarrow\mathfrak{X}(\mathcal{B})$, and the divergence operator $\mathrm{div}:\mathfrak{X}(\mathcal{B}) \rightarrow C^{\infty}(\mathcal{B})$. It is easy to show that $\mathrm{curl}\circ\mathrm{grad}=0$, and $\mathrm{div}\circ\mathrm{curl}=0$. These relations allow one to write the following complex    
%-----------------------------
\begin{equation}
\scalebox{1}{\xymatrix@C=3ex{
0\ar[r] &C^{\infty}(\mathcal{B}) \ar[r]^-{\mathrm{grad}} &\mathfrak{X}(\mathcal{B}) \ar[r]^-{\mathrm{curl}} & \mathfrak{X}(\mathcal{B}) \ar[r]^-{\mathrm{div}} &C^{\infty}(\mathcal{B}) \ar[r]  &0,  }} 
\end{equation}
%----------------------------- 
that is called the grad-curl-div complex or simply the \textsf{gcd} complex. It turns out that the \textsf{gcd} complex is equivalent to the de Rham complex in the following sense. Let $\{X^{I}\}$ be the Cartesian coordinates of $\mathbb{R}^{3}$. We can define the following isomorphisms  
%-----------------------------
\begin{equation}\label{isogdc}
\begin{aligned}
\imath_{0}&:C^{\infty}(\mathcal{B})\rightarrow\Omega^{0}(\mathcal{B}), &&\imath_{0}(f)=f,  \\
\imath_{1}&:\mathfrak{X}(\mathcal{B})\rightarrow\Omega^{1}(\mathcal{B}),&&\left(\imath_{1}(\boldsymbol{Y})\right)_{I}=Y^{I},  \\
\imath_{2}&:\mathfrak{X}(\mathcal{B})\rightarrow\Omega^{2}(\mathcal{B}),&&\left(\imath_{2}(\boldsymbol{Y})\right)_{IJ}=\varepsilon_{IJK}Y^{K},  \\
\imath_{3}&:C^{\infty}(\mathcal{B})\rightarrow\Omega^{3}(\mathcal{B}), &&\left(\imath_{3}(f)\right)_{123}=f,    
\end{aligned}
\end{equation}
%-----------------------------  
where $\varepsilon_{IJK}$ is the standard permutation symbol. Simple calculations show that 
%-----------------------------
\begin{equation}%\label{IsoComplexes}
\imath_{1}\circ\mathrm{grad}=d\circ\imath_{0},~~~~ \imath_{2}\circ\mathrm{curl}=d\circ\imath_{1}, 
~~~~\imath_{3}\circ\mathrm{div}=d\circ\imath_{2}. \nonumber
\end{equation}
%----------------------------- 
These relations can be succinctly depicted by the following diagram. 
%-----------------------------
\begin{equation}
\begin{gathered}\label{gcdDeRhamCom}
\scalebox{1}{\xymatrix@C=3ex{0 \ar[r] &C^{\infty}(\mathcal{B}) \ar[r]^-{\mathrm{grad}} \ar[d]^{\imath_{0}} &\mathfrak{X}(\mathcal{B}) \ar[r]^-{\mathrm{curl}} \ar[d]^{\imath_{1}} &\mathfrak{X}(\mathcal{B}) \ar[r]^-{\mathrm{div}} \ar[d]^{\imath_{2}} &C^{\infty}(\mathcal{B})\ar[r] \ar[d]^{\imath_{3}} &0  \\
0\ar[r] &\Omega^{0}(\mathcal{B}) \ar[r]^-{d} &\Omega^{1}(\mathcal{B}) \ar[r]^-{d} & \Omega^{2}(\mathcal{B}) \ar[r]^-{d} &\Omega^{3}(\mathcal{B}) \ar[r] &0  } }
\end{gathered}
\end{equation}
%----------------------------- 
Diagram (\ref{gcdDeRhamCom}) suggests that any result holding for the de Rham complex should have a counterpart for the \textsf{gcd} complex as well.\footnote{More precisely, isomorphisms $\imath_{0},\dots,\imath_{3}$ induce a complex isomorphism.} For example, diagram (\ref{gcdDeRhamCom}) implies that $\imath_{0},\dots,\imath_{3}$ also induce isomorphisms between the cohomology groups. This means that $\boldsymbol{Y}=\mathrm{grad}\,f$, if and only if $\imath_{1}(\boldsymbol{Y})=d(\imath_{0}(f))$, and similarly $\boldsymbol{Y}=\mathrm{curl}\,\boldsymbol{Z}$, if and only if $\imath_{2}(\boldsymbol{Y})=d(\imath_{1}(\boldsymbol{Z}))$. By using the de Rham theorem and (\ref{ExactForm}), one can show that $\boldsymbol{Y}$ is the gradient of a function if and only if   
%-----------------------------
\begin{equation}\label{Exactgrad}
\mathrm{curl}\,\boldsymbol{Y}=0, \text{ and }~~~ \int_{\ell}\imath_{1}(\boldsymbol{Y})=\int_{\ell}\boldsymbol{G}(\boldsymbol{Y},\boldsymbol{t}_{\ell})dS=0, ~\forall\ell\subset\mathcal{B},
\end{equation}
%----------------------------- 
where $\ell$ is an arbitrary closed curve in $\mathcal{B}$, $\boldsymbol{t}_{\ell}$ is the unit tangent vector field along $\ell$, and $\boldsymbol{G}(\boldsymbol{Y},\boldsymbol{t}_{\ell})$ is the standard inner product of $\boldsymbol{Y}$ and $\boldsymbol{t}_{\ell}$ in $\mathbb{R}^{3}$. Similarly, one concludes that $\boldsymbol{Y}$ is the curl of a vector field if and only if   
%-----------------------------
\begin{equation}\label{Exactcurl}
\mathrm{div}\,\boldsymbol{Y}=0, \text{ and }~~~ \int_{\mathcal{C}}\imath_{2}(\boldsymbol{Y})=\int_{\mathcal{C}}\boldsymbol{G}(\boldsymbol{Y},\boldsymbol{N}_{\!\mathcal{C}})dA=0, ~\forall \mathcal{C}\subset\mathcal{B}, 
\end{equation}
%----------------------------- 
where $\mathcal{C}$ is any closed surface in $\mathcal{B}$ and $\boldsymbol{N}_{\!\mathcal{C}}$ is the unit outward normal vector field of $\mathcal{C}$. If $\bar{\mathcal{B}}$ is compact and if we restrict $C^{\infty}(\mathcal{B})$ and $\mathfrak{X}(\mathcal{B})$ to smooth functions and vector fields over $\bar{\mathcal{B}}$, then by equipping the spaces in diagram (\ref{gcdDeRhamCom}) with appropriate $L^{2}$-inner products, $\imath_{0},\dots,\imath_{3}$ become isometries. Therefore, any orthogonal decomposition for $\Omega^{k}(\bar{\mathcal{B}})$, $k=1,2$, induces an equivalent decomposition for $\mathfrak{X}(\bar{\mathcal{B}})$ as well, e.g. see \citet[Corollary 3.5.2]{Schwarz1995}. Moreover, one can study solutions of the vector Laplacian $\boldsymbol{\Delta}=\mathrm{grad}\circ\mathrm{div}-\mathrm{curl}\circ\mathrm{curl}$, and develop stable numerical schemes for it by using the corresponding results for the Hodge-Laplacian \citep{Arnold2006,Arnold2010}. In summary, the diagram (\ref{gcdDeRhamCom}) allows one to extend all the standard results developed for the de Rham complex to the \textsf{gcd} complex.

The notion of a complex has been extensively used in linear elasticity. Motivated by the mechanics of distributed defects, and in particular incompatibility of plastic strains, \citet{Kroner1959} introduced the linear elasticity complex, also called the Kr\"{o}ner complex, which is equivalent to a complex in differential geometry due to \citet{Calabi1961}. \citet{Eastwood2000} derived a construction of the linear elasticity complex from the de Rham complex. \citet{Arnold2006} used the linear elasticity complex and obtained the first stable mixed formulation for linear elasticity. This complex can be used for deriving Hodge-type decompositions for linear elasticity as well \citep{GeymonatKrasucki2009}. To our best knowledge, there has not been any discussion on analogous differential complexes for general (nonlinear) continua, and in particular, nonlinear elasticity.

\paragraph{Contributions of this paper.}Introducing differential complexes for general continua is the main goal of this paper. We can summarize the main contributions as follows.
\begin{itemize}
\item We show that a tensorial analogue of the $\mathsf{gcd}$ complex called the $\boldsymbol{\mathsf{GCD}}$ complex, can describe both the kinematics and the kinetics of motions of continua. More specifically, the $\boldsymbol{\mathsf{GCD}}$ complex involves the displacement gradient and the first Piola-Kirchhoff stress tensor. We show that a diagram similar to (\ref{gcdDeRhamCom}) commutes for the $\boldsymbol{\mathsf{GCD}}$ complex as well, and therefore, the nonlinear compatibility equations in terms of the displacement gradient and the existence of stress functions for the first Piola-Kirchhoff stress tensor directly follow from (\ref{ExactForm}). Another tensorial version of the $\mathsf{gcd}$ complex is the $\boldsymbol{\mathsf{gcd}}$ complex that involves non-symmetric second-order tensors. This complex allows one to introduce stress functions for non-symmetric Cauchy stress and the second Piola-Kirchhoff stress tensors. By using the Cauchy and the second Piola-Kirchhoff stresses, one obtains complexes that only describe the kinetics of motion.

\item It has been mentioned in several references in the literature that the linear elasticity complex is equivalent to the Calabi complex, e.g. see \citep{Eastwood2000}. Although this equivalence is trivial for the kinematics part of the linear elasticity complex, in our opinion, it is not trivial at all for the kinetics part. Therefore, we include a discussion on the equivalence of these complexes by using a diagram similar to (\ref{gcdDeRhamCom}). Another reason for studying the above equivalence is that it helps us understand the relation between the linear elasticity complex and the $\boldsymbol{\mathsf{GCD}}$ complex. In particular, the linear elasticity complex is not the linearization of the $\boldsymbol{\mathsf{GCD}}$ complex. The Calabi complex also provides a coordinate-free expression for the linear compatibility equations. Using the above complexes one observes that on a $3$-manifold, the linear and nonlinear compatibility problems, and the existence of stress functions are related to $\mathrm{H}^{1}_{dR}(\mathcal{B})$ and $\mathrm{H}^{2}_{dR}(\mathcal{B})$, respectively.
 
\item Using the ideas underlying the Calabi complex, we derive the nonlinear compatibility equations in terms of the Green deformation tensor for motions of bodies (with the same dimensions as ambient spaces) and shells in curved ambient spaces with constant curvatures.

\end{itemize}

\paragraph{Notation.}In this paper, we use the pair of smooth Riemannian manifolds $(\mathcal{B},\boldsymbol{G})$ with local coordinates $\{X^{I}\}$ and $(\mathcal{S},\boldsymbol{g})$ with local coordinates $\{x^{i}\}$ to denote a general continuum and its ambient space, respectively. If $\mathcal{B}\subset\mathbb{R}^{n}$, then $\bar{\mathcal{B}}$ denotes the closure of $\mathcal{B}$ in $\mathbb{R}^{n}$. Unless stated otherwise, we assume the summation convention on repeated indices. The space of smooth real-valued functions on $\mathcal{B}$ is denoted by $C^{\infty}(\mathcal{B})$. We use $\Gamma(\mathcal{V})$ to indicate smooth sections of a vector bundle $\mathcal{V}$. Thus, $\Gamma(\text{\large{$\otimes$}}^{2}T\mathcal{B})$ and $\Gamma(\text{\large{$\otimes$}}^{2}T^{\ast}\mathcal{B})$ are the spaces of $({}^{2}_{0})$- and $({}^{0}_{2})$-tensors on $\mathcal{B}$. The space of symmetric $({}^{0}_{2})$-tensors is denoted by $\Gamma(S^{2}T^{\ast}\mathcal{B})$. It is customary to write $\mathfrak{X}(\mathcal{B}):=\Gamma(T\mathcal{B})$, and $\Omega^{k}(\mathcal{B}):=\Gamma(\Lambda^{k}T^{\ast}\mathcal{B})$, i.e. $\Omega^{k}(\mathcal{B})$ is the space of anti-symmetric $({}^{0}_{k})$-tensors or simply $k$-forms. Tensors are indicated by bold letters, e.g. $\boldsymbol{T}\in\Gamma(\text{\large{$\otimes$}}^{2}T\mathcal{B})$ and its components are denoted by $T^{IJ}$ or $(\boldsymbol{T})^{IJ}$. The space of $k$-forms with values in $\mathbb{R}^{n}$ is denoted by $\Omega^{k}(\mathcal{B};\mathbb{R}^{n})$, i.e. if $\boldsymbol{\alpha}\in\Omega^{k}(\mathcal{B};\mathbb{R}^{n})$, and $\mathbf{X}_{1},\dots,\mathbf{X}_{k}\in T_{X}\mathcal{B}$, then $\boldsymbol{\alpha}(\mathbf{X}_{1},\dots,\mathbf{X}_{k})\in \mathbb{R}^{n}$, and $\boldsymbol{\alpha}$ is anti-symmetric. Let $\varphi:\mathcal{B}\rightarrow\mathcal{S}$ be a smooth mapping. The space of two-point tensors over $\varphi$ with components $F^{iI}$ is denoted by $\Gamma(T\varphi(\mathcal{B})\otimes T\mathcal{B})$.

%----------------------------- 
%----------------------------- 
\section{Differential Complexes for Second-Order Tensors}

In this section, we study some differential complexes for $2$D and $3$D flat manifolds that contain second-order tensors. These complexes fall into two categories: Those induced by the de Rham complex and those induced by the Calabi complex.  Complexes induced by the de Rham complex include arbitrary second-order tensors and can be considered as tensorial versions of the $\mathsf{gcd}$ complex. Complexes induced by the Calabi complex involve only symmetric second-order tensors. In \S3, we study the applications of these complexes to some classical problems in continuum mechanics.

%----------------------------- 
%----------------------------- 
\subsection{Complexes Induced by the de Rham Complex}

Complexes for second-order tensors that are induced by the de Rham complex only contain first-order differential operators. We begin our discussion by considering $3$-manifolds and will later study $2$-manifolds separately.

%----------------------------- 
%----------------------------- 
\subsubsection{Complexes for flat $3$-manifolds}

Let $\mathcal{B}\subset\mathbb{R}^{3}$ be an open subset and suppose $\{X^{I}\}$ is the Cartesian coordinates on $\mathcal{B}$. We equip $\mathcal{B}$ with metric $\boldsymbol{G}$, which is the Euclidean metric of $\mathbb{R}^{3}$. The gradient of vector fields and the curl and the divergence of $({}^{2}_{0})$-tensors are defined as
%-----------------------------
\begin{align}%\label{TensorOpe}
\mathbf{grad}&:\mathfrak{X}(\mathcal{B})\rightarrow\Gamma(\text{\large{$\otimes$}}^{2}T\mathcal{B}), &&(\mathbf{grad}\,\boldsymbol{Y})^{IJ}=Y^{I}{}_{,J}, \nonumber\\
\mathbf{curl}&:\Gamma(\text{\large{$\otimes$}}^{2}T\mathcal{B})\rightarrow\Gamma(\text{\large{$\otimes$}}^{2}T\mathcal{B}), &&(\mathbf{curl}\,\boldsymbol{T})^{IJ}=\varepsilon_{IKL}T^{JL}{}_{,K}, \nonumber \\
\mathbf{div}&:\Gamma(\text{\large{$\otimes$}}^{2}T\mathcal{B})\rightarrow\mathfrak{X}(\mathcal{B}), &&(\mathbf{div}\,\boldsymbol{T})^{I}=T^{IJ}{}_{,J}, \nonumber   
\end{align}
%-----------------------------  
where ``$_{,J}$'' indicates $\partial/\partial X^{J}$. We also define the operator 
%-----------------------------
\begin{equation}
\mathbf{curl}^{\mathsf{T}}:\Gamma(\text{\large{$\otimes$}}^{2}T\mathcal{B})\rightarrow\Gamma(\text{\large{$\otimes$}}^{2}T\mathcal{B}), ~(\mathbf{curl}^{\mathsf{T}}\boldsymbol{T})^{IJ}=(\mathbf{curl}\,\boldsymbol{T})^{JI}. \nonumber
\end{equation}
%-----------------------------  
It is straightforward to show that $\mathbf{curl}^{\mathsf{T}}\circ\mathbf{grad}=0$, and $\mathbf{div}\circ\mathbf{curl}^{\mathsf{T}}=0$. Thus, we obtain the following complex
%-----------------------------
\begin{equation}
\scalebox{1}{\xymatrix@C=3ex{
0\ar[r] &\mathfrak{X}(\mathcal{B}) \ar[r]^-{\mathbf{grad}} &\Gamma(\text{\large{$\otimes$}}^{2}T\mathcal{B}) \ar[r]^-{\mathbf{curl}^{\mathsf{T}}} & \Gamma(\text{\large{$\otimes$}}^{2}T\mathcal{B}) \ar[r]^-{\mathbf{div}} &\mathfrak{X}(\mathcal{B}) \ar[r]  &0,  }} 
\end{equation}
%----------------------------- 
that, due to its resemblance with the $\mathsf{gcd}$ complex, is called the $\boldsymbol{\mathsf{gcd}}$ complex. Interestingly, similar to the $\mathsf{gcd}$ complex, useful properties of the $\boldsymbol{\mathsf{gcd}}$ complex also follow from the de Rham complex. This can be described via the $\mathbb{R}^{3}$-valued de Rham complex as follows. Let $d:\Omega^{k}(\mathcal{B})\rightarrow\Omega^{k+1}(\mathcal{B})$ be the standard exterior derivative given by
%-----------------------------
\begin{equation}\label{LocExDer}
(d\boldsymbol{\beta})_{I_{0}\cdots I_{k}}=\sum_{i=0}^{k}(-1)^{i}\beta_{I_{0}\cdots \widehat{I}_{i}\cdots I_{k},I_{i}}, \nonumber  
\end{equation}
%----------------------------- 
where the hat over an index implies the elimination of that index. Any $\boldsymbol{\alpha}\in\Omega^{k}(\mathcal{B};\mathbb{R}^{3})$ can be considered as $\boldsymbol{\alpha}=(\boldsymbol{\alpha}^{1},\boldsymbol{\alpha}^{2},\boldsymbol{\alpha}^{3})$, with $\boldsymbol{\alpha}^{i}\in\Omega^{k}(\mathcal{B})$, $i=1,2,3$. One can define the exterior derivative $\boldsymbol{d}:\Omega^{k}(\mathcal{B};\mathbb{R}^{3})\rightarrow\Omega^{k+1}(\mathcal{B};\mathbb{R}^{3})$ by $\boldsymbol{d}\boldsymbol{\alpha}=(d\boldsymbol{\alpha}^{1},d\boldsymbol{\alpha}^{2},d\boldsymbol{\alpha}^{3})$. Since $d\circ d =0$, we also conclude that $\boldsymbol{d}\circ\boldsymbol{d}=0$, which leads to the $\mathbb{R}^{3}$-valued de Rham complex $\left( \Omega(\mathcal{B};\mathbb{R}^{3}),\boldsymbol{d} \right)$. Given $\boldsymbol{\alpha}\in\Omega^{k}(\mathcal{B};\mathbb{R}^{3})$, let $[\boldsymbol{\alpha}]^{i}{}_{I_{1}\cdots I_{k}}$ denote the components of $\boldsymbol{\alpha}^{i}\in\Omega^{k}(\mathcal{B})$. By using the global orthonormal coordinate system $\{X^{I}\}$, one can define the following isomorphisms
%-----------------------------
\begin{align}%\label{TensorIso}
\boldsymbol{\imath}_{0}&:\mathfrak{X}(\mathcal{B})\rightarrow\Omega^{0}(\mathcal{B};\mathbb{R}^{3}), & &[\boldsymbol{\imath}_{0}(\boldsymbol{Y})]^{i}=\text{\textdelta}_{iI}Y^{I}, \nonumber\\
\boldsymbol{\imath}_{1}&:\Gamma(\text{\large{$\otimes$}}^{2}T\mathcal{B})\rightarrow\Omega^{1}(\mathcal{B};\mathbb{R}^{3}), & &[\boldsymbol{\imath}_{1}(\boldsymbol{T})]^{i}{}_{J}=\text{\textdelta}_{iI}T^{IJ}, \nonumber \\
\boldsymbol{\imath}_{2}&:\Gamma(\text{\large{$\otimes$}}^{2}T\mathcal{B})\rightarrow\Omega^{2}(\mathcal{B};\mathbb{R}^{3}), & &[\boldsymbol{\imath}_{2}(\boldsymbol{T})]^{i}{}_{JK}=\text{\textdelta}_{iI}\varepsilon_{JKL}T^{IL}, \nonumber \\
\boldsymbol{\imath}_{3}&:\mathfrak{X}(\mathcal{B})\rightarrow\Omega^{3}(\mathcal{B};\mathbb{R}^{3}), & &[\boldsymbol{\imath}_{3}(\boldsymbol{Y})]^{i}{}_{123}=\text{\textdelta}_{iI}Y^{I}, \nonumber 
\end{align}
%-----------------------------  
where $\text{\textdelta}_{iI}$ is the Kronecker delta. Let $\boldsymbol{T}^{\mathsf{T}}$ be the transpose of $\boldsymbol{T}$, i.e. $\left(\boldsymbol{T}^{\mathsf{T}}\right)^{IJ}= T^{JI}$, and let $\{\mathbf{E}_{I}\}$ be the standard basis of $\mathbb{R}^{3}$. For $\boldsymbol{T}\in\Gamma(\text{\large{$\otimes$}}^{2}T\mathcal{B})$, we define $\overrightarrow{\boldsymbol{T}}_{\!\!\mathbf{N}}$ to be the traction of $\boldsymbol{T}^{\mathsf{T}}$ in the direction of unit vector $\mathbf{N}=N^{I}\mathbf{E}_{I}\in\EuScript{S}^{2}$, where $\EuScript{S}^{2}\subset\mathbb{R}^{3}$ is the unit $2$-sphere. Thus, $\overrightarrow{\boldsymbol{T}}_{\!\!\mathbf{N}}=N^{I}T^{IJ}\mathbf{E}_{J}$. By using (\ref{isogdc}), we can write 
%-----------------------------
\begin{equation}\label{isocompo}
\boldsymbol{\imath}_{k}(\boldsymbol{T})= \left(\imath_{k}\!\!\left(\overrightarrow{\boldsymbol{T}}_{\!\!\mathbf{E}_{1}}\right) , \imath_{k}\!\!\left(\overrightarrow{\boldsymbol{T}}_{\!\!\mathbf{E}_{2}}\right) , \imath_{k}\!\!\left(\overrightarrow{\boldsymbol{T}}_{\!\!\mathbf{E}_{3}}\right)  \right), ~k=1,2.  
\end{equation}
%----------------------------- 
It is easy to show that
%-----------------------------
\begin{equation}\label{IsoComplexesTenVecVal}
\boldsymbol{\imath}_{1}\circ\mathbf{grad}=\boldsymbol{d}\circ\boldsymbol{\imath}_{0},~ \boldsymbol{\imath}_{2}\circ\mathbf{curl}^{\mathsf{T}}=\boldsymbol{d}\circ\boldsymbol{\imath}_{1}, ~\boldsymbol{\imath}_{3}\circ\mathbf{div}=\boldsymbol{d}\circ\boldsymbol{\imath}_{2}. \nonumber 
\end{equation}
%----------------------------- 
Therefore, the following diagram, which is the tensorial analogue of the diagram (\ref{gcdDeRhamCom}) commutes for the $\boldsymbol{\mathsf{gcd}}$ complex.
%-----------------------------
\begin{equation}
\begin{gathered}\label{gcddeRh}
\scalebox{1}{\xymatrix@C=3ex{0 \ar[r] &\mathfrak{X}(\mathcal{B}) \ar[r]^-{\mathbf{grad}} \ar[d]^{\boldsymbol{\imath}_{0}} &\Gamma(\text{\large{$\otimes$}}^{2}T\mathcal{B}) \ar[r]^-{\mathbf{curl}^{\!\mathsf{T}}} \ar[d]^{\boldsymbol{\imath}_{1}} &\Gamma(\text{\large{$\otimes$}}^{2}T\mathcal{B}) \ar[r]^-{\mathbf{div}} \ar[d]^{\boldsymbol{\imath}_{2}} &\mathfrak{X}(\mathcal{B}) \ar[r] \ar[d]^{\boldsymbol{\imath}_{3}} &0  \\
0\ar[r] &\Omega^{0}(\mathcal{B};\mathbb{R}^{3}) \ar[r]^-{\boldsymbol{d}} & \Omega^{1}(\mathcal{B};\mathbb{R}^{3}) \ar[r]^-{\boldsymbol{d}} &\Omega^{2}(\mathcal{B};\mathbb{R}^{3}) \ar[r]^-{\boldsymbol{d}} & \Omega^{3}(\mathcal{B};\mathbb{R}^{3}) \ar[r] &0 } } 
\end{gathered}
\end{equation}
%-----------------------------

\begin{rem} Diagram (\ref{gcdDeRhamCom}) is valid for any $3$-manifold, see \citet[\S 3.5]{Schwarz1995} for the definitions of $\mathrm{grad}$, $\mathrm{curl}$, and $\mathrm{div}$ on arbitrary $3$-manifolds. However, we require a global orthonormal coordinate system for defining $\mathbf{curl}$ and the isomorphisms $\boldsymbol{\imath}_{k}$. Thus, the $\boldsymbol{\mathsf{gcd}}$ complex and the diagram (\ref{gcddeRh}) are valid merely on flat $3$-manifolds.
\end{rem}
  
The contraction $\langle \boldsymbol{T},\boldsymbol{Y}\rangle$ of $\boldsymbol{T}\in\Gamma(\text{\large{$\otimes$}}^{2}T\mathcal{B})$ and $\boldsymbol{Y}\in\mathfrak{X}(\mathcal{B})$ is a vector field that in the orthonormal coordinate system $\{X^{I}\}$ reads $\langle \boldsymbol{T},\boldsymbol{Y}\rangle=T^{IJ}Y^{J}\mathbf{E}_{I}$. Clearly, if $\boldsymbol{N}_{\!\mathcal{C}}$ is the unit outward normal vector field of a closed surface $\mathcal{C}\subset\mathcal{B}$, then $\langle \boldsymbol{T},\boldsymbol{N}_{\!\mathcal{C}}\rangle$ is the traction of $\boldsymbol{T}$ on $\mathcal{C}$. Suppose $\mathrm{H}^{k}_{\boldsymbol{\mathsf{gcd}}}(\mathcal{B})$ is the $k$-th cohomology group of the $\boldsymbol{\mathsf{gcd}}$ complex. Diagram (\ref{gcddeRh}) implies that $\boldsymbol{\imath}_{k}$ also induces the isomorphism $\mathrm{H}^{k}_{\boldsymbol{\mathsf{gcd}}}(\mathcal{B})\approx \text{\large{$\oplus$}}^{3}_{i=1}\mathrm{H}^{k}_{dR}(\mathcal{B})$ between the cohomology groups. Using this fact and (\ref{Exactgrad}), we can prove the following theorem.

\begin{thm}\label{ExactgradTen} An arbitrary tensor $\boldsymbol{T}\in\Gamma(\text{\large{$\otimes$}}^{2}T\mathcal{B})$ is the gradient of a vector field if and only if 
%-----------------------------
\begin{equation}\label{ExactgradTenEq}
\mathbf{curl}^{\mathsf{T}}\boldsymbol{T}=0, \text{ and }~~~\int_{\ell}\langle\boldsymbol{T},\boldsymbol{t}_{\ell}\rangle dS=0, ~\forall\ell\subset\mathcal{B}, 
\end{equation}
%----------------------------- 
where $\ell$ is an arbitrary closed curve in $\mathcal{B}$ and $\boldsymbol{t}_{\ell}$ is the unit tangent vector field along $\ell$. 
\end{thm}

\begin{proof} By using (\ref{isocompo}) and diagram (\ref{gcddeRh}), we conclude that $\boldsymbol{T}=\mathbf{grad}\,\boldsymbol{Y}$, if and only if $\boldsymbol{\imath}_{1}(\boldsymbol{T})=\boldsymbol{d}(\boldsymbol{\imath}_{0}(\boldsymbol{Y}))$, if and only if $\imath_{1}\left(\overrightarrow{\boldsymbol{T}}_{\!\!\mathbf{E}_{I}}\right)=d\, Y^{I}$, $I=1,2,3$. The condition (\ref{Exactgrad}) implies that in addition to $\mathbf{curl}^{\mathsf{T}}\boldsymbol{T}=0$, $\boldsymbol{T}$ should also satisfy   
%-----------------------------
\begin{equation}\label{ExactgradTenEq2}
\int_{\ell}\imath_{1}\!\!\left(\overrightarrow{\boldsymbol{T}}_{\!\!\mathbf{E}_{I}}\right)=\int_{\ell}\boldsymbol{G}(\overrightarrow{\boldsymbol{T}}_{\!\!\mathbf{E}_{I}},\boldsymbol{t}_{\ell})dS=0, ~\forall\ell\subset\mathcal{B}, ~I=1,2,3, \nonumber
\end{equation}
%----------------------------- 
which is equivalent to the integral condition in (\ref{ExactgradTenEq}). 
\end{proof}

Similarly, one can use (\ref{Exactcurl}) for deriving the necessary and sufficient conditions for the existence of a potential for $\boldsymbol{T}$ induced by $\mathbf{curl}^{\mathsf{T}}$. The upshot is the following theorem.  

\begin{thm}\label{ExactcurlTTen} Given $\boldsymbol{T}\in\Gamma(\text{\large{$\otimes$}}^{2}T\mathcal{B})$, there exists $\boldsymbol{W}\in\Gamma(\text{\large{$\otimes$}}^{2}T\mathcal{B})$ such that $\boldsymbol{T}=\mathbf{curl}^{\mathsf{T}}\boldsymbol{W}$, if and only if 
%-----------------------------
\begin{equation}\label{ExactcurlTTenEq}
\mathbf{div}\,\boldsymbol{T}=0, \text{ and }~~~\int_{\mathcal{C}}\langle\boldsymbol{T},\boldsymbol{N}_{\!\mathcal{C}}\rangle dA=0, ~\forall \mathcal{C} \subset\mathcal{B},
\end{equation}
%----------------------------- 
where $\mathcal{C}$ is an arbitrary closed surface in $\mathcal{B}$ and $\boldsymbol{N}_{\!\mathcal{C}}$ is its unit outward normal vector field. 
\end{thm}

\begin{rem}\label{IndNumCo}If the Betti numbers $b_{k}(\mathcal{B})$, $k=1,2$, are finite, then it suffices to check (\ref{ExactgradTenEq}) and (\ref{ExactcurlTTenEq}) for $b_{1}(\mathcal{B})$ and $b_{2}(\mathcal{B})$ ``independent'' closed curves and closed surfaces, respectively. In particular, one concludes that if $\mathcal{B}$ is simply-connected, then any $\mathbf{curl}^{\mathsf{T}}$-free $({}^{2}_{0})$-tensor is the gradient of a vector field and if $\mathcal{B}$ is contractible, then any $\mathbf{div}$-free $({}^{2}_{0})$-tensor admits a $\mathbf{curl}^{\mathsf{T}}$-potential. If $\bar{\mathcal{B}}\subset\mathbb{R}^{3}$ is compact, i.e. $\bar{\mathcal{B}}$ is closed and bounded, then all $b_{k}(\mathcal{B})$'s are finite. The calculation of $b_{k}(\mathcal{B})$ for some physically interesting bodies and the selection of independent closed loops and closed surfaces are discussed in \citep{AngoshtariYavari2014II}.
\end{rem}

We can also write an analogue of the $\boldsymbol{\mathsf{gcd}}$ complex for two-point tensors. Let $\mathcal{S}=\mathbb{R}^{3}$ with coordinate system $\{x^{i}\}$, which is the Cartesian coordinates of $\mathbb{R}^{3}$. Suppose $\varphi:\mathcal{B}\rightarrow\mathcal{S}$ is a smooth mapping and let $T_{X}\varphi(\mathcal{B}):=T_{\varphi(X)}\mathcal{S}$. Note that although $\varphi$ is not necessarily an embedding, the dimension of $T_{X}\varphi(\mathcal{B})$ is always equal to $\dim\mathcal{S}$. We can define the following operators for two-point tensors that belong to $\Gamma(T\varphi(\mathcal{B}))$ and $\Gamma(T\varphi(\mathcal{B})\otimes T\mathcal{B})$:
%-----------------------------
\begin{align}
\mathbf{Grad}&:\Gamma(T\varphi(\mathcal{B}))\rightarrow\Gamma(T\varphi(\mathcal{B})\otimes T\mathcal{B}), &&(\mathbf{Grad}\,\boldsymbol{U})^{iI}=U^{i}{}_{,I}, \nonumber \\  
\mathbf{Curl}^{\mathsf{T}}&:\Gamma(T\varphi(\mathcal{B})\otimes T\mathcal{B})\rightarrow\Gamma(T\varphi(\mathcal{B})\otimes T\mathcal{B}), &&(\mathbf{Curl}^{\mathsf{T}}\boldsymbol{F})^{iI}=\varepsilon_{IKL}F^{iL}{}_{,K}, \nonumber\\
\mathbf{Div}&:\Gamma(T\varphi(\mathcal{B})\otimes T\mathcal{B})\rightarrow\Gamma(T\varphi(\mathcal{B})), &&(\mathbf{Div}\,\boldsymbol{F})^{i}=F^{iI}{}_{,I}. \nonumber  
\end{align}
%-----------------------------  
We have $\mathbf{Curl}^{\mathsf{T}}\circ\mathbf{Grad}=0$, and $\mathbf{Div}\circ\mathbf{Curl}^{\mathsf{T}}=0$. Thus, the $\boldsymbol{\mathsf{GCD}}$ complex is written as:
%-----------------------------
\begin{equation}
\scalebox{.97}{\xymatrix@C=3ex{
0\ar[r] &\Gamma(T\varphi(\mathcal{B})) \ar[r]^-{\mathbf{Grad}} &\Gamma(T\varphi(\mathcal{B})\otimes T\mathcal{B}) \ar[r]^-{\mathbf{Curl}^{\mathsf{T}}} & \Gamma(T\varphi(\mathcal{B})\otimes T\mathcal{B}) \ar[r]^-{\mathbf{Div}} &\Gamma(T\varphi(\mathcal{B})) \ar[r]  &0.  }} \nonumber
\end{equation}
%----------------------------- 
By using the following isomorphisms
%-----------------------------
\begin{align}
\boldsymbol{I}_{0}&:\Gamma(T\varphi(\mathcal{B}))\rightarrow\Omega^{0}(\mathcal{B};\mathbb{R}^{3}), & &[\boldsymbol{I}_{0}(\boldsymbol{U})]^{i}=U^{i}, \nonumber\\
\boldsymbol{I}_{1}&:\Gamma(T\varphi(\mathcal{B})\otimes T\mathcal{B})\rightarrow \Omega^{1}(\mathcal{B};\mathbb{R}^{3}), & &[\boldsymbol{I}_{1}(\boldsymbol{F})]^{i}{}_{J}=F^{iJ}, \nonumber \\   
\boldsymbol{I}_{2}&:\Gamma(T\varphi(\mathcal{B})\otimes T\mathcal{B})\rightarrow \Omega^{2}(\mathcal{B};\mathbb{R}^{3}), & &[\boldsymbol{I}_{2}(\boldsymbol{F})]^{i}{}_{JK}=\varepsilon_{JKL}F^{iL}, \nonumber \\ 
\boldsymbol{I}_{3}&:\Gamma(T\varphi(\mathcal{B}))\rightarrow\Omega^{3}(\mathcal{B};\mathbb{R}^{3}), & &[\boldsymbol{I}_{3}(\boldsymbol{U})]^{i}{}_{123}=U^{i}, \nonumber
\end{align}
%-----------------------------  
one concludes that the following diagram commutes.  
%-----------------------------
\begin{equation}
\begin{gathered}\label{gcdTdeRh}
\scalebox{.97}{\xymatrix@C=3ex{0 \ar[r] &\Gamma(T\varphi(\mathcal{B})) \ar[r]^-{\mathbf{Grad}} \ar[d]^{\boldsymbol{I}_{0}} &\Gamma(T\varphi(\mathcal{B})\otimes T\mathcal{B}) \ar[r]^-{\mathbf{Curl}^{\!\mathsf{T}}} \ar[d]^{\boldsymbol{I}_{1}} &\Gamma(T\varphi(\mathcal{B})\otimes T\mathcal{B}) \ar[r]^-{\mathbf{Div}} \ar[d]^{\boldsymbol{I}_{2}} &\Gamma(T\varphi(\mathcal{B})) \ar[r] \ar[d]^{\boldsymbol{I}_{3}} &0  \\
0\ar[r] &\Omega^{0}(\mathcal{B};\mathbb{R}^{3}) \ar[r]^-{\boldsymbol{d}} & \Omega^{1}(\mathcal{B};\mathbb{R}^{3}) \ar[r]^-{\boldsymbol{d}} &\Omega^{2}(\mathcal{B};\mathbb{R}^{3}) \ar[r]^-{\boldsymbol{d}} & \Omega^{3}(\mathcal{B};\mathbb{R}^{3}) \ar[r] &0 } } \nonumber 
\end{gathered}
\end{equation}
%-----------------------------
The above isomorphisms also induce an isomorphism $\mathrm{H}^{k}_{\boldsymbol{\mathsf{GCD}}}(\mathcal{B})\approx \text{\large{$\oplus$}}^{3}_{i=1}\mathrm{H}^{k}_{dR}(\mathcal{B})$, where $\mathrm{H}^{k}_{\boldsymbol{\mathsf{GCD}}}(\mathcal{B})$ is the $k$-th cohomology group of the $\boldsymbol{\mathsf{GCD}}$ complex. Let $\{\mathbf{E}_{I}\}$ and $\{\mathbf{e}_{i}\}$ be two copies of the standard basis of $\mathbb{R}^{3}$. For $\boldsymbol{F}\in\Gamma(T\varphi(\mathcal{B})\otimes T\mathcal{B})$, and $\mathbf{n}=n^{i}\mathbf{e}_{i}\in\EuScript{S}^{2}$, let $\overrightarrow{\boldsymbol{F}}_{\!\!\mathbf{n}}=n^{i}F^{iJ}\mathbf{E}_{J}\in\mathfrak{X}(\mathcal{B})$. Then, one can write
%-----------------------------
\begin{equation}
\boldsymbol{I}_{k}(\boldsymbol{F})= \left(\imath_{k}\!\!\left(\overrightarrow{\boldsymbol{F}}_{\!\!\mathbf{e}_{1}}\right) , \imath_{k}\!\!\left(\overrightarrow{\boldsymbol{F}}_{\!\!\mathbf{e}_{2}}\right) , \imath_{k}\!\!\left(\overrightarrow{\boldsymbol{F}}_{\!\!\mathbf{e}_{3}}\right)  \right), ~k=1,2. \nonumber  
\end{equation}
%----------------------------- 
Let $\langle\boldsymbol{F},\boldsymbol{Y}\rangle:=F^{iI}Y^{I}\mathbf{e}_{i}$. The above relations for the $\boldsymbol{\mathsf{GCD}}$ complex allow us to obtain the following results that can be proved similarly to Theorems \ref{ExactgradTen} and \ref{ExactcurlTTen}.

\begin{thm}\label{ExactGradCurlTPTen} Given $\boldsymbol{F}\in\Gamma(T\varphi(\mathcal{B})\otimes T\mathcal{B})$, there exists $\boldsymbol{U}\in\Gamma(T\varphi(\mathcal{B}))$ such that $\boldsymbol{F}=\mathbf{Grad}\,\boldsymbol{U}$, if and only if
%-----------------------------
\begin{equation}\label{ExactGradCurlTenEq1}
\mathbf{Curl}^{\mathsf{T}}\boldsymbol{F}=0, \text{ and }~~~\int_{\ell}\langle\boldsymbol{F},\boldsymbol{t}_{\ell}\rangle dS=0, ~\forall\ell\subset\mathcal{B}. \nonumber 
\end{equation}
%----------------------------- 
Moreover, there exists $\boldsymbol{\Psi}\in\Gamma(T\varphi(\mathcal{B})\otimes T\mathcal{B})$ such that $\boldsymbol{F}=\mathbf{Curl}^{\mathsf{T}}\boldsymbol{\Psi}$, if and only if 
%-----------------------------
\begin{equation}\label{ExactGradCurlTenEq2}
\mathbf{Div}\,\boldsymbol{F}=0, \text{ and }~~~\int_{\mathcal{C}}\langle\boldsymbol{F},\boldsymbol{N}_{\!\mathcal{C}}\rangle dA=0, ~\forall \mathcal{C} \subset\mathcal{B}. \nonumber
\end{equation}
%----------------------------- 
\end{thm}

\begin{rem} Note that for writing the $\boldsymbol{\mathsf{GCD}}$ complex, only $\mathcal{S}$ needs to be flat and admit a global orthonormal coordinate system. This observation is useful for deriving a complex for motions of $2$D surfaces (shells) in $\mathbb{R}^{3}$. By using the natural isomorphism $\flat:\mathfrak{X}(\mathcal{B})\rightarrow\Omega^{1}(\mathcal{B})$ induced by $\boldsymbol{G}$ and the Hodge star operator $\ast:\Omega^{k}(\mathcal{B})\rightarrow\Omega^{n-k}(\mathcal{B})$, where $n=\dim\mathcal{B}$, we can write
%-----------------------------
\begin{equation}
\begin{aligned}
\boldsymbol{I}_{1}(\boldsymbol{F})&=\big( (\boldsymbol{F}(dx^{1}))^{\flat} , (\boldsymbol{F}(dx^{2}))^{\flat} , (\boldsymbol{F}(dx^{3}))^{\flat} \big), \nonumber \\
\boldsymbol{I}_{2}(\boldsymbol{F})&=\big( \ast(\boldsymbol{F}(dx^{1}))^{\flat} , \ast(\boldsymbol{F}(dx^{2}))^{\flat} , \ast(\boldsymbol{F}(dx^{3}))^{\flat} \big).
\end{aligned}
\end{equation}
%----------------------------- 
\end{rem}

%----------------------------- 
%----------------------------- 
\subsubsection{Complexes for $2$-manifolds}

Let $\mathcal{B}\subset\mathbb{R}^{2}$ be a $2$-manifold and suppose $\{X^{I}\}$ is the Cartesian coordinate system. For $2$-manifolds, instead of $\mathbf{curl}^{\mathsf{T}}$, we define the operator
%-----------------------------
\begin{equation}%\label{TensorOpe2D}
\mathbf{c}:\Gamma(\text{\large{$\otimes$}}^{2}T\mathcal{B})\rightarrow\mathfrak{X}(\mathcal{B}), ~(\mathbf{c}(\boldsymbol{T}))^{I}=T^{I2}{}_{,1}-T^{I1}{}_{,2}, \nonumber   
\end{equation}
%-----------------------------  
that satisfies $\mathbf{c}\circ\mathbf{grad}=0$. Also consider the following isomorphisms
%-----------------------------
\begin{align}%\label{TensorIso2D}
\boldsymbol{\jmath}_{0}&:\mathfrak{X}(\mathcal{B})\rightarrow\Omega^{0}(\mathcal{B};\mathbb{R}^{2}), & &[\boldsymbol{\jmath}_{0}(\boldsymbol{Y})]^{i}=\text{\textdelta}_{iI}Y^{I}, \nonumber\\
\boldsymbol{\jmath}_{1}&:\Gamma(\text{\large{$\otimes$}}^{2}T\mathcal{B})\rightarrow\Omega^{1}(\mathcal{B};\mathbb{R}^{2}), & &[\boldsymbol{\jmath}_{1}(\boldsymbol{T})]^{i}{}_{J}=\text{\textdelta}_{iI}T^{IJ}, \nonumber \\
\boldsymbol{\jmath}_{2}&:\mathfrak{X}(\mathcal{B})\rightarrow\Omega^{2}(\mathcal{B};\mathbb{R}^{2}), & &[\boldsymbol{\jmath}_{2}(\boldsymbol{Y})]^{i}{}_{12}=\text{\textdelta}_{iI}Y^{I}. \nonumber 
\end{align}
%-----------------------------  
It is straightforward to show that the following diagram commutes.
%-----------------------------
\begin{equation}
\begin{gathered}\label{gcdeRh}
\scalebox{1}{\xymatrix@C=3ex{0 \ar[r] &\mathfrak{X}(\mathcal{B}) \ar[r]^-{\mathbf{grad}} \ar[d]^{\boldsymbol{\jmath}_{0}} &\Gamma(\text{\large{$\otimes$}}^{2}T\mathcal{B}) \ar[r]^-{\mathbf{c}} \ar[d]^{\boldsymbol{\jmath}_{1}} &\mathfrak{X}(\mathcal{B}) \ar[r] \ar[d]^{\boldsymbol{\jmath}_{2}} &0  \\
0\ar[r] &\Omega^{0}(\mathcal{B};\mathbb{R}^{2}) \ar[r]^-{\boldsymbol{d}} & \Omega^{1}(\mathcal{B};\mathbb{R}^{2}) \ar[r]^-{\boldsymbol{d}} &\Omega^{2}(\mathcal{B};\mathbb{R}^{2}) \ar[r] &0 } } 
\end{gathered}
\end{equation}
%-----------------------------
The complex in the first row of (\ref{gcdeRh}) is called the $\boldsymbol{\mathsf{gc}}$ complex. This diagram implies that $\mathrm{H}^{k}_{\boldsymbol{\mathsf{gc}}}(\mathcal{B})\approx \text{\large{$\oplus$}}^{2}_{i=1}\mathrm{H}^{k}_{dR}(\mathcal{B})$, where $\mathrm{H}^{k}_{\boldsymbol{\mathsf{gc}}}(\mathcal{B})$ is the $k$-th cohomology group of the $\boldsymbol{\mathsf{gc}}$ complex and we obtain the following result.

\begin{thm}\label{ExactgradTen2D} A tensor $\boldsymbol{T}\in\Gamma(\text{\large{$\otimes$}}^{2}T\mathcal{B})$ on a $2$-manifold $\mathcal{B}\subset\mathbb{R}^{2}$ is the gradient of a vector field if and only if 
%-----------------------------
\begin{equation}\label{ExactgradTen2DEq}
\mathbf{c}(\boldsymbol{T})=0, \text{ and }~~~\int_{\ell}\langle\boldsymbol{T},\boldsymbol{t}_{\ell}\rangle dS=0, ~\forall\ell\subset\mathcal{B}. \nonumber
\end{equation}
%----------------------------- 
\end{thm}

For $2$-manifolds, we can write a second complex that contains $\mathbf{div}$. In an orthonormal coordinate system $\{X^{I}\}$, the codifferential operator $\delta_{k}:\Omega^{k}(\mathcal{B})\rightarrow\Omega^{k-1}(\mathcal{B})$ reads
%-----------------------------
\begin{equation}\label{LocCoDif}
(\delta\boldsymbol{\beta})_{I_{1}\cdots I_{k-1}}=-\beta_{JI_{1}\cdots  I_{k-1},J}. \nonumber  
\end{equation}
%----------------------------- 
We have $\delta\circ\delta=0$, that gives rise to the complex $\left(\Omega(\mathcal{B}),\delta\right)$ with the cohomology groups $\mathrm{H}^{k}_{co}(\mathcal{B}):=\ker\delta_{k}/\mathrm{im}\,\delta_{k+1}$. Using the Hodge star operator $\ast:\Omega^{k}(\mathcal{B})\rightarrow\Omega^{n-k}(\mathcal{B})$, it is straightforward to show that $\mathrm{H}^{k}_{co}(\mathcal{B})\approx\mathrm{H}^{n-k}_{dR}(\mathcal{B})$. One can also write the complex $\left(\Omega(\mathcal{B};\mathbb{R}^{2}),\boldsymbol{\delta}\right)$, where $\boldsymbol{\delta}\boldsymbol{\alpha}=(\delta\boldsymbol{\alpha}^{1},\delta\boldsymbol{\alpha}^{2})$. By defining the operator  
%-----------------------------
\begin{equation}
\mathbf{s}:\mathfrak{X}(\mathcal{B})\rightarrow\Gamma(\text{\large{$\otimes$}}^{2}T\mathcal{B}), ~(\mathbf{s}(\boldsymbol{Y}))^{IJ}=\text{\textdelta}^{1J}Y^{I}{}_{,2}-\text{\textdelta}^{2J}Y^{I}{}_{,1}, \nonumber   
\end{equation}
%-----------------------------  
we obtain the following diagram. 
%-----------------------------
\begin{equation}
\begin{gathered}\label{sdCdeRh}
\scalebox{1}{\xymatrix@C=3ex{0 &\mathfrak{X}(\mathcal{B}) \ar[l]  \ar[d]^{-\boldsymbol{\jmath}_{0}} &\Gamma(\text{\large{$\otimes$}}^{2}T\mathcal{B}) \ar[l]_-{\mathbf{div}} \ar[d]^{\boldsymbol{\jmath}_{1}} &\mathfrak{X}(\mathcal{B}) \ar[l]_-{\mathbf{s}}  \ar[d]^{\boldsymbol{\jmath}_{2}} &0 \ar[l]  \\
0 &\Omega^{0}(\mathcal{B};\mathbb{R}^{2}) \ar[l] & \Omega^{1}(\mathcal{B};\mathbb{R}^{2}) \ar[l]_-{\boldsymbol{\delta}} &\Omega^{2}(\mathcal{B};\mathbb{R}^{2}) \ar[l]_-{\boldsymbol{\delta}} &0 \ar[l] } }  
\end{gathered}
\end{equation}
%-----------------------------
We call the first row of (\ref{sdCdeRh}) the $\boldsymbol{\mathsf{sd}}$ complex and denote its homology groups by $\mathrm{H}^{k}_{\boldsymbol{\mathsf{sd}}}(\mathcal{B})$. We have $\mathrm{H}^{k}_{\boldsymbol{\mathsf{sd}}}(\mathcal{B})\approx \text{\large{$\oplus$}}^{2}_{i=1}\mathrm{H}^{n-k}_{dR}(\mathcal{B})$. Let $\{\mathbf{E}_{I}\}$ be the standard basis of $\mathbb{R}^{2}$ and let $\boldsymbol{N}_{\!\ell}$ be a unit vector field along a closed curve $\ell$, which is normal to the tangent vector field $\boldsymbol{t}_{\ell}$, such that $\{\boldsymbol{t}_{\ell},\boldsymbol{N}_{\!\ell}\}$ has the same orientation as $\{\mathbf{E}_{1},\mathbf{E}_{2}\}$ does. The following theorem is the analogue of Theorem \ref{ExactgradTen2D} for the $\boldsymbol{\mathsf{sd}}$ complex.

\begin{thm}\label{ExactsTen2D} On a $2$-manifold $\mathcal{B}\subset\mathbb{R}^{2}$, there exists $\boldsymbol{Y}\in\mathfrak{X}(\mathcal{B})$ for $\boldsymbol{T}\in\Gamma(\text{\large{$\otimes$}}^{2}T\mathcal{B})$ such that $\boldsymbol{T}=\mathbf{s}(\boldsymbol{Y})$, if and only if 
%-----------------------------
\begin{equation}\label{ExactsTen2DEq}
\mathbf{div}\,\boldsymbol{T}=0, \text{ and }~~~\int_{\ell}\langle\boldsymbol{T},\boldsymbol{N}_{\!\ell}\rangle dS=0, ~\forall\ell\subset\mathcal{B}. 
\end{equation}
%----------------------------- 
\end{thm}

\begin{proof} We know that $\boldsymbol{T}=\mathbf{s}(\boldsymbol{Y})$, if and only if $\boldsymbol{\jmath}_{1}(\boldsymbol{T})=\boldsymbol{\delta}(\boldsymbol{\jmath}_{2}(\boldsymbol{Y}))$, if and only if $(\overrightarrow{\boldsymbol{T}}_{\!\!\mathbf{E}_{I}})^{\flat}=\delta Y^{I}$, $I=1,2$. The Hodge star operator induces an isomorphism between the cohomology groups of $\left(\Omega(\mathcal{B}),d\right)$ and $\left(\Omega(\mathcal{B}),\delta\right)$, and therefore, the last condition is equivalent to $\ast(\overrightarrow{\boldsymbol{T}}_{\!\!\mathbf{E}_{I}})^{\flat}=(\overrightarrow{\boldsymbol{T}}^{\perp}_{\!\!\mathbf{E}_{I}})^{\flat}=d Y^{I}$, where $\overrightarrow{\boldsymbol{T}}^{\perp}_{\!\!\mathbf{E}_{I}}=(T^{I2},-T^{I1})$, $I=1,2$. Since $\boldsymbol{G}(\overrightarrow{\boldsymbol{T}}^{\perp}_{\!\!\mathbf{E}_{I}},\boldsymbol{t}_{\ell})=\boldsymbol{G}(\overrightarrow{\boldsymbol{T}}_{\!\!\mathbf{E}_{I}},\boldsymbol{N}_{\ell})$, one obtains (\ref{ExactsTen2DEq}).  
\end{proof}

Next, suppose $\varphi:\mathcal{B}\rightarrow\mathbb{R}^{2}$ is a smooth mapping and let $\{x^{i}\}$ be the Cartesian coordinates of $\mathbb{R}^{2}$ with $\{\mathbf{e}_{i}\}$ being its standard basis. Consider the following isomorphisms
%-----------------------------
\begin{align}
\boldsymbol{J}_{0}&:\Gamma(T\varphi(\mathcal{B}))\rightarrow\Omega^{0}(\mathcal{B};\mathbb{R}^{2}), & &[\boldsymbol{J}_{0}(\boldsymbol{U})]^{i}=U^{i}, \nonumber\\
\boldsymbol{J}_{1}&:\Gamma(T\varphi(\mathcal{B})\otimes T\mathcal{B})\rightarrow \Omega^{1}(\mathcal{B};\mathbb{R}^{2}), & &[\boldsymbol{J}_{1}(\boldsymbol{F})]^{i}{}_{J}=F^{iJ}, \nonumber \\   
\boldsymbol{J}_{2}&:\Gamma(T\varphi(\mathcal{B}))\rightarrow\Omega^{2}(\mathcal{B};\mathbb{R}^{2}), & &[\boldsymbol{J}_{2}(\boldsymbol{U})]^{i}{}_{12}=U^{i}, \nonumber
\end{align}
%-----------------------------  
together with the operators
%-----------------------------
\begin{align}
\mathbf{C}&:\Gamma(T\varphi(\mathcal{B})\otimes T\mathcal{B})\rightarrow\Gamma(T\varphi(\mathcal{B})), &&(\mathbf{C}(\boldsymbol{F}))^{i}=F^{i2}{}_{,1}-F^{i1}{}_{,2}, \nonumber  \\ 
\mathbf{S}&:\Gamma(T\varphi(\mathcal{B}))\rightarrow\Gamma(T\varphi(\mathcal{B})\otimes T\mathcal{B}), &&(\mathbf{S}(\boldsymbol{U}))^{iI}=\text{\textdelta}^{1I}U^{i}{}_{,2}-\text{\textdelta}^{2I}U^{i}{}_{,1}. \nonumber   
\end{align}
%-----------------------------  
Replacing $\boldsymbol{\jmath}_{0}$, $\boldsymbol{\jmath}_{1}$, $\boldsymbol{\jmath}_{2}$, $\mathbf{c}$, and $\mathbf{s}$ with $\boldsymbol{J}_{0}$, $\boldsymbol{J}_{1}$, $\boldsymbol{J}_{2}$, $\mathbf{C}$, and $\mathbf{S}$, respectively, in diagrams (\ref{gcdeRh}) and (\ref{sdCdeRh}) gives us the corresponding diagrams for two-point tensors. The associated complexes are called the $\boldsymbol{\mathsf{GC}}$ and the $\boldsymbol{\mathsf{SD}}$ complexes and we have the following result.

\begin{thm}\label{ExactGradS2pTen2D} Let $\varphi:\mathcal{B}\rightarrow\mathbb{R}^{2}$ be a smooth mapping and $\boldsymbol{F}\in\Gamma(T\varphi(\mathcal{B})\otimes T\mathcal{B})$. We have $\boldsymbol{F}=\mathbf{Grad}\,\boldsymbol{U}$, if and only if
%-----------------------------
\begin{equation}\label{ExactGrad2pTen2DEq}
\mathbf{C}(\boldsymbol{F})=0, \text{ and }~~~\int_{\ell}\langle\boldsymbol{F},\boldsymbol{t}_{\ell}\rangle dS=0, ~\forall\ell\subset\mathcal{B}. \nonumber
\end{equation}
%----------------------------- 
Moreover, we can write $\boldsymbol{F}=\mathbf{S}(\boldsymbol{U})$, if and only if 
%-----------------------------
\begin{equation}\label{ExactS2pTen2DEq}
\mathbf{Div}\,\boldsymbol{F}=0, \text{ and }~~~\int_{\ell}\langle\boldsymbol{F},\boldsymbol{N}_{\!\ell}\rangle dS=0, ~\forall\ell\subset\mathcal{B}. \nonumber
\end{equation}
%----------------------------- 
\end{thm}

As was mentioned earlier, the complexes for two-point tensors do not require $\mathcal{B}$ to be flat. This allows one to obtain a complex describing motions of $2$D surfaces (shells) in $\mathbb{R}^{3}$. Let $(\mathcal{B},\boldsymbol{G})$ be a $2$D surface in $\mathbb{R}^{3}$ with an arbitrary local coordinate system $\{X^{I}\}$, $I=1,2$, and let $\{x^{i}\}$ and $\{\mathbf{e}_{i}\}$, $i=1,2,3$, be the Cartesian coordinates and the standard basis of $\mathbb{R}^{3}$, respectively. The local basis for $T\mathcal{B}$ induced by $\{X^{I}\}$ is denoted by $\{\boldsymbol{E}_{I}\}$. Suppose $\varphi:\mathcal{B}\rightarrow\mathbb{R}^{3}$ is a smooth mapping and consider the following isomorphisms   
%-----------------------------
\begin{align}
\mathbf{J}_{0}&:\Gamma(T\varphi(\mathcal{B}))\rightarrow\Omega^{0}(\mathcal{B};\mathbb{R}^{3}), & &[\mathbf{J}_{0}(\boldsymbol{U})]^{i}=U^{i}, \nonumber\\
\mathbf{J}_{1}&:\Gamma(T\varphi(\mathcal{B})\otimes T\mathcal{B})\rightarrow \Omega^{1}(\mathcal{B};\mathbb{R}^{3}), & &[\mathbf{J}_{1}(\boldsymbol{F})]^{i}{}_{J}=G_{JI}F^{iI}, \nonumber \\   
\mathbf{J}_{2}&:\Gamma(T\varphi(\mathcal{B}))\rightarrow\Omega^{2}(\mathcal{B};\mathbb{R}^{3}), & &[\mathbf{J}_{2}(\boldsymbol{U})]^{i}{}_{12}=\sqrt{\det G_{IJ}}\,U^{i}, \nonumber
\end{align}
%-----------------------------  
where $G_{IJ}$ are the components of $\boldsymbol{G}$ and $\det G_{IJ}$ is the determinant of the matrix $[G_{IJ}]_{2\times2}$. Let $G^{IJ}$ be the components of the inverse of $[G_{IJ}]_{2\times2}$. We define the operators $\mathsf{Grad}:\Gamma(T\varphi(\mathcal{B}))\rightarrow\Gamma(T\varphi(\mathcal{B})\otimes T\mathcal{B})$ and $\mathsf{C}:\Gamma(T\varphi(\mathcal{B})\otimes T\mathcal{B})\rightarrow\Gamma(T\varphi(\mathcal{B}))$ by 
%-----------------------------  
\begin{equation}
(\mathsf{Grad}\,\boldsymbol{U})^{iI}=G^{IJ}U^{i}{}_{,J}, ~~  
(\mathsf{C}(\boldsymbol{F}))^{i}=\frac{\left( G_{2K}F^{iK}\right)_{,1}-\left(G_{1K}F^{iK}\right)_{,2}}{\sqrt{\det G_{IJ} }}. \nonumber    
\end{equation}
%-----------------------------  
Using the above operators, one obtains the following diagram for the \emph{\textsf{GC}} complex.
%-----------------------------
\begin{equation}
\begin{gathered}\label{GCSerdeRh}
\scalebox{1}{\xymatrix@C=3ex{0 \ar[r] &\Gamma(T\varphi(\mathcal{B})) \ar[r]^-{\mathsf{Grad}} \ar[d]^{\mathbf{J}_{0}} &\Gamma(T\varphi(\mathcal{B})\otimes T\mathcal{B}) \ar[r]^-{\mathsf{C}} \ar[d]^{\mathbf{J}_{1}} &\Gamma(T\varphi(\mathcal{B})) \ar[r] \ar[d]^{\mathbf{J}_{2}} &0  \\
0\ar[r] &\Omega^{0}(\mathcal{B};\mathbb{R}^{3}) \ar[r]^-{\boldsymbol{d}} & \Omega^{1}(\mathcal{B};\mathbb{R}^{3}) \ar[r]^-{\boldsymbol{d}} &\Omega^{2}(\mathcal{B};\mathbb{R}^{3}) \ar[r] &0 } } \nonumber
\end{gathered}
\end{equation}
%-----------------------------
Thus, the following result holds.

\begin{thm}\label{ExactGradSSurf2pTen2D} Let $\mathcal{B}$ be a $2$D surface and let $\varphi:\mathcal{B}\rightarrow\mathbb{R}^{3}$ be a smooth mapping. Then, $\boldsymbol{F}\in\Gamma(T\varphi(\mathcal{B})\otimes T\mathcal{B})$ can be written as $\boldsymbol{F}=\mathsf{Grad}\,\boldsymbol{U}$, if and only if
%-----------------------------
\begin{equation}\label{ExactGradSSurf2pTen2DEq}
\mathsf{C}(\boldsymbol{F})=0, \text{ and }~~~\int_{\ell}\langle\boldsymbol{F},\boldsymbol{t}_{\ell}\rangle dS=0, ~\forall\ell\subset\mathcal{B}, \nonumber
\end{equation}
%----------------------------- 
where $\langle\boldsymbol{F},\boldsymbol{t}_{\ell}\rangle= G_{IJ}F^{iJ}(\boldsymbol{t}_{\ell})^{I}\mathbf{e}_{i}$.
\end{thm}

%----------------------------- 
%----------------------------- 
\subsection{Complexes Induced by the Calabi Complex}

A differential complex suitable for symmetric second-order tensors was introduced by \citet{Calabi1961}. It is well-known that the Calabi complex in $\mathbb{R}^{3}$ is equivalent to the linear elasticity complex \citep{Eastwood2000}. In this section, we study the Calabi complex and its connection with the linear elasticity complex in some details. As we will see later, this study provides a framework for writing the nonlinear compatibility equations in curved ambient spaces and comparing stress functions induced by the Calabi complex with those induced by the $\boldsymbol{\mathsf{gcd}}$ or the $\boldsymbol{\mathsf{GCD}}$ complexes. Moreover, the Calabi complex provides a coordinate-free expression for the linear compatibility equations.

The Calabi complex is valid on any Riemannian manifold with constant (sectional) curvature (also called a Clifford-Klein space). These spaces are defined as follows. Let $\nabla$ be the Levi-Civita connection of $(\mathcal{B},\boldsymbol{G})$ and let $\boldsymbol{X}_{i}\in\mathfrak{X}(\mathcal{B})$, $i=1,\dots,5$. The curvature $\boldsymbol{R}$ and the Riemannian curvature $\boldsymbol{\mathcal{R}}$ induced by $\boldsymbol{G}$ are given by $\boldsymbol{R}(\boldsymbol{X}_{\!1},\boldsymbol{X}_{\!2})\boldsymbol{X}_{\!3}=\nabla_{\boldsymbol{X}_{\!1}}\!\nabla_{\boldsymbol{X}_{\!2}}\boldsymbol{X}_{\!3}-\nabla_{\boldsymbol{X}_{\!2}}\!\nabla_{\boldsymbol{X}_{\!1}}\boldsymbol{X}_{\!3}-\nabla_{[\boldsymbol{X}_{\!1},\boldsymbol{X}_{\!2}]}\boldsymbol{X}_{\!3}$, and $\boldsymbol{\mathcal{R}}(\boldsymbol{X}_{\!1},\boldsymbol{X}_{\!2},\boldsymbol{X}_{\!3},\boldsymbol{X}_{\!4})=\boldsymbol{G}(\boldsymbol{R}(\boldsymbol{X}_{\!1},\boldsymbol{X}_{\!2})\boldsymbol{X}_{\!3},\boldsymbol{X}_{\!4})$. Let $\boldsymbol{\Sigma}_{X}$ be a 2-dimensional subspace of $T_{X}\mathcal{B}$ and let $\mathbf{X}_{1},\mathbf{X}_{2}\in \boldsymbol{\Sigma}_{X}$ be two arbitrary linearly independent vectors. The sectional curvature of $\boldsymbol{\Sigma}_{X}$ is defined as  
%-----------------------------
\begin{equation}
K(\boldsymbol{\Sigma}_{X})=\frac{\boldsymbol{\mathcal{R}}(\mathbf{X}_{1},\mathbf{X}_{2},\mathbf{X}_{2},\mathbf{X}_{1})}{\left(\boldsymbol{G}(\mathbf{X}_{1},\mathbf{X}_{1})\boldsymbol{G}(\mathbf{X}_{2},\mathbf{X}_{2})\right)^{2}-\left(\boldsymbol{G}(\mathbf{X}_{1},\mathbf{X}_{2})\right)^{2}}. \nonumber
\end{equation}
%-----------------------------  
Sectional curvature $K(\boldsymbol{\Sigma}_{X})$ is independent of the choice of $\mathbf{X}_{1}$ and $\mathbf{X}_{2}$ \citep{Docarmo1992}. A manifold $\mathcal{B}$ has a constant curvature $\mathrm{k}\in\mathbb{R}$ if and only if $K(\boldsymbol{\Sigma}_{X})=\mathrm{k}$, $\forall X\in\mathcal{B}$ and $\forall \boldsymbol{\Sigma}_{X}\subset T_{X}\mathcal{B}$. If $\mathcal{B}$ is complete and simply-connected, it is isometric to: (i) the $n$-sphere with radius $1/\sqrt{\mathrm{k}}$, if $\mathrm{k}>0$, (ii) $\mathbb{R}^{n}$, if $\mathrm{k}=0$, and (iii) the hyperbolic space, if $\mathrm{k}<0$ \citep{KN1963}. An arbitrary Riemannian manifold with constant curvature is locally isometric to one of the above manifolds depending on the sign of $\mathrm{k}$. For example, the sectional curvature of a cylinder in $\mathbb{R}^{3}$ is zero and the cylinder is locally isometric to $\mathbb{R}^{2}$. Discussions on the classification of Riemannian manifolds with constant curvatures can be found in \citet{Wolf2011}. One can show that $(\mathcal{B},\boldsymbol{G})$ has constant curvature $\mathrm{k}$ if and only if    
%-----------------------------
\begin{equation}\label{com_con_sct}
\boldsymbol{R}(\boldsymbol{X}_{\!1},\boldsymbol{X}_{\!2})\boldsymbol{X}_{\!3}=\mathrm{k}\big(\boldsymbol{G}(\boldsymbol{X}_{\!3},\boldsymbol{X}_{\!2})\boldsymbol{X}_{\!1} - \boldsymbol{G}(\boldsymbol{X}_{\!3},\boldsymbol{X}_{\!1})\boldsymbol{X}_{\!2}\big). 
\end{equation}
%-----------------------------    
Similar to the de Rham complex, the Calabi complex on $n$-manifolds terminates after $n$ non-trivial operators. For $n=3$, these operators are: the Killing operator $\mathrm{D}_{0}$, the linearized curvature operator $\mathrm{D}_{1}$, and the Bianchi operator $\mathrm{D}_{2}$.

The first operator in the Calabi complex on $(\boldsymbol{B},\boldsymbol{G})$ is the Killing operator $\mathrm{D}_{0}:\mathfrak{X}(\mathcal{B})\rightarrow\Gamma(S^{2}T^{\ast}\mathcal{B})$ defined as
%-----------------------------
\begin{equation}%\label{lin_st_sub}
(\mathrm{D}_{0}\boldsymbol{U})(\boldsymbol{X}_{\!1},\boldsymbol{X}_{\!2})= \frac{1}{2}\Big(\boldsymbol{G}(\boldsymbol{X}_{\!1},\nabla_{\boldsymbol{X}_{\!2}}\boldsymbol{U}) + \boldsymbol{G}(\nabla_{\boldsymbol{X}_{\!1}}\boldsymbol{U},\boldsymbol{X}_{\!2})\Big).\nonumber
\end{equation}
%-----------------------------
Note that $\mathrm{D}_{0}\boldsymbol{U}=\frac{1}{2}\EuScript{L}_{\boldsymbol{U}}\boldsymbol{G}$, where $\EuScript{L}_{\boldsymbol{U}}$ is the Lie derivative. The kernel of $\mathrm{D}_{0}$ coincides with the space of Killing vector fields $\text{\textTheta}(\mathcal{B})$ on $\mathcal{B}$.\footnote{A Killing vector field $\boldsymbol{U}\in\text{\textTheta}(\mathcal{B})$ is also called an infinitesimal isometry in the sense that its flow $\mathrm{Fl}^{\boldsymbol{U}}$ induces an isometry $\mathrm{Fl}^{\boldsymbol{U}}_{t}:=\mathrm{Fl}^{\boldsymbol{U}}(t,\cdot):U\subset\mathcal{B}\rightarrow\mathcal{B}$ \citep{Docarmo1992}.} If an $n$-manifold $\mathcal{B}$ is a subset of $\mathbb{R}^{n}$ with Cartesian coordinates $\{X^{I}\}$, any $\boldsymbol{U}\in \text{\textTheta}(\mathcal{B})$ at $X=(X^{1},\dots,X^{n})\in\mathbb{R}^{n}$ can be written as $\boldsymbol{U}(X)=\mathbf{v}+A\cdot X$, where $\mathbf{v}\in\mathbb{R}^{n}$, $A\in\mathfrak{so}(\mathbb{R}^{n}):=\{A\in\mathbb{R}^{n\times n}:A+A^{\mathsf{T}}=0\}$, with $\mathbb{R}^{n\times n}$ being the space of real $n\times n$ matrices.\footnote{This implies that $\text{\textTheta}(\mathcal{B})$ is isomorphic to $\mathfrak{euc}(\mathbb{R}^{n})$, which is the Lie algebra of the group of rigid body motions $Euc(\mathbb{R}^{n})$.} Therefore, we conclude that $\dim \text{\textTheta}(\mathcal{B})=n(n+1)/2$.
 
The second operator of the Calabi complex can be obtained by linearizing the Riemannian curvature. Let $\boldsymbol{A}$ be a Riemannian metric on $\mathcal{B}$ and let $\nabla^{\boldsymbol{A}}$ and $\boldsymbol{\mathcal{R}}^{\boldsymbol{A}}$ be the corresponding Levi-Civita connection and Riemannian curvature, respectively. The tensor $\boldsymbol{\mathcal{R}}^{\boldsymbol{A}}$ has the following symmetries.
%-----------------------------
%\begin{subequations}
%\begin{align}
%\boldsymbol{\mathcal{R}}^{\boldsymbol{A}}(\!\boldsymbol{X}_{1},\boldsymbol{X}_{2},\boldsymbol{X}_{3},\boldsymbol{X}_{4}) &+ \boldsymbol{\mathcal{R}}^{\boldsymbol{A}}(\boldsymbol{X}_{2},\boldsymbol{X}_{3},\boldsymbol{X}_{1},\boldsymbol{X}_{4}) + 
%\boldsymbol{\mathcal{R}}^{\boldsymbol{A}}(\boldsymbol{X}_{3},\boldsymbol{X}_{1},\boldsymbol{X}_{2},\boldsymbol{X}_{4})=0, \label{Bia}\\
%\boldsymbol{\mathcal{R}}^{\boldsymbol{A}}(\boldsymbol{X}_{1},\boldsymbol{X}_{2},\boldsymbol{X}_{3},\boldsymbol{X}_{4}) &= -\boldsymbol{\mathcal{R}}^{\boldsymbol{A}}(\boldsymbol{X}_{2},\boldsymbol{X}_{1},\boldsymbol{X}_{3},\boldsymbol{X}_{4}) = 
%-\boldsymbol{\mathcal{R}}^{\boldsymbol{A}}(\boldsymbol{X}_{1},\boldsymbol{X}_{2},\boldsymbol{X}_{4},\boldsymbol{X}_{3}). \label{SymRie2} 
%\end{align}
%\end{subequations}
%-----------------------------
\begin{equation}\label{Bia}
\begin{aligned}
& \boldsymbol{\mathcal{R}}^{\boldsymbol{A}}(\!\boldsymbol{X}_{1},\boldsymbol{X}_{2},\boldsymbol{X}_{3},\boldsymbol{X}_{4}) + \boldsymbol{\mathcal{R}}^{\boldsymbol{A}}(\boldsymbol{X}_{2},\boldsymbol{X}_{3},\boldsymbol{X}_{1},\boldsymbol{X}_{4}) \\
& + 
\boldsymbol{\mathcal{R}}^{\boldsymbol{A}}(\boldsymbol{X}_{3},\boldsymbol{X}_{1},\boldsymbol{X}_{2},\boldsymbol{X}_{4})=0,
\end{aligned}
\end{equation}
%-----------------------------
\begin{equation}\label{SymRie2} 
\begin{aligned}
& \boldsymbol{\mathcal{R}}^{\boldsymbol{A}}(\boldsymbol{X}_{1},\boldsymbol{X}_{2},\boldsymbol{X}_{3},\boldsymbol{X}_{4}) = -\boldsymbol{\mathcal{R}}^{\boldsymbol{A}}(\boldsymbol{X}_{2},\boldsymbol{X}_{1},\boldsymbol{X}_{3},\boldsymbol{X}_{4}) \\
&= 
-\boldsymbol{\mathcal{R}}^{\boldsymbol{A}}(\boldsymbol{X}_{1},\boldsymbol{X}_{2},\boldsymbol{X}_{4},\boldsymbol{X}_{3}). 
\end{aligned}
\end{equation}
%-----------------------------
Equivalently, the components $\mathcal{R}^{\boldsymbol{A}}_{I_{1}I_{2}I_{3}I_{4}}$ of $\boldsymbol{\mathcal{R}}^{\boldsymbol{A}}$ satisfy
%-----------------------------
\begin{equation}
\begin{aligned}
\mathcal{R}^{\boldsymbol{A}}_{I_{1}I_{2}I_{3}I_{4}}&+\mathcal{R}^{\boldsymbol{A}}_{I_{2}I_{3}I_{1}I_{4}}+\mathcal{R}^{\boldsymbol{A}}_{I_{3}I_{1}I_{2}I_{4}}=0, \\   \nonumber
\mathcal{R}^{\boldsymbol{A}}_{I_{1}I_{2}I_{3}I_{4}}&=-\mathcal{R}^{\boldsymbol{A}}_{I_{2}I_{1}I_{3}I_{4}}=-\mathcal{R}^{\boldsymbol{A}}_{I_{1}I_{2}I_{4}I_{3}}.
\end{aligned}
\end{equation}
%-----------------------------
The identity (\ref{Bia}) is called the first Bianchi identity. The above symmetries imply that $\boldsymbol{\mathcal{R}}^{\boldsymbol{A}}$ has $n^{2}(n^{2}-1)/12$ independent components \citep{Srivastava2008}. The relations (\ref{Bia}) and (\ref{SymRie2}) also induce the symmetry 
%-----------------------------
\begin{equation}
\boldsymbol{\mathcal{R}}^{\boldsymbol{A}}(\boldsymbol{X}_{\!1},\boldsymbol{X}_{\!2},\boldsymbol{X}_{\!3},\boldsymbol{X}_{\!4}) = \boldsymbol{\mathcal{R}}^{\boldsymbol{A}}(\boldsymbol{X}_{\!3},\boldsymbol{X}_{\!4},\boldsymbol{X}_{\!1},\boldsymbol{X}_{\!2}),\label{SymRie3}
\end{equation}
%-----------------------------  
i.e. $\mathcal{R}^{\boldsymbol{A}}_{I_{1}I_{2}I_{3}I_{4}}=\mathcal{R}^{\boldsymbol{A}}_{I_{3}I_{4}I_{1}I_{2}}$. For $n=2,3$, (\ref{SymRie2}) and (\ref{SymRie3}) determine all the symmetries of $\boldsymbol{\mathcal{R}}^{\boldsymbol{A}}$, and therefore, the space of tensors with the symmetries of the Riemannian curvature is $\Gamma(S^{2}(\Lambda^{2}T^{\ast}\mathcal{B}))$.\footnote{Tensors in $\Gamma(S^{2}(\Lambda^{2}T^{\ast}\mathcal{B}))$ have $(n^{2}-n+2)(n^{2}-n)/8$ independent components. For $n\geq4$, (\ref{SymRie2}) and (\ref{SymRie3}) do not imply (\ref{Bia}), and thus, tensors with the symmetries of the Riemannian curvature belong to a subspace of $\Gamma(S^{2}(\Lambda^{2}T^{\ast}\mathcal{B}))$. If $T^{\ast}\mathcal{B}$ is induced by a representation, i.e. it is a homogeneous vector bundle corresponding to an irreducible representation, the representation theory provides some tools to specify tensors with complicated symmetries such as those of the Riemannian curvature \citep{BastonEastwood1989,PenroseRindler1984}.} Let $\boldsymbol{e}\in\Gamma(S^{2}T^{\ast}\mathcal{B})$ be an arbitrary symmetric $({}^{0}_{2})$-tensor. The linearization of the operator $\boldsymbol{A}\mapsto\boldsymbol{\mathcal{R}}^{\boldsymbol{A}}$ is the linear operator $\mathbf{r}^{\boldsymbol{A}}:\Gamma(S^{2}T^{\ast}\mathcal{B})\rightarrow \Gamma(S^{2}(\Lambda^{2}T^{\ast}\mathcal{B}))$ defined by $\mathbf{r}^{\boldsymbol{A}}(\boldsymbol{e}):=\frac{d}{dt}\big|_{t=0}\boldsymbol{\mathcal{R}}^{\boldsymbol{A}+t\boldsymbol{e}}$ \citep{GaGo2004,GaGo1983}. One can write       
%-----------------------------
\begin{align}
	 2\mathbf{r}^{\boldsymbol{A}}(\boldsymbol{e})(&\boldsymbol{X}_1,\boldsymbol{X}_2,\boldsymbol{X}_3,\boldsymbol{X}_4)
	=  \mathbf{L}^{\boldsymbol{A}}(\boldsymbol{e})(\boldsymbol{X}_1,\boldsymbol{X}_2,\boldsymbol{X}_3,\boldsymbol{X}_4)   
	\nonumber \\
	&+ \boldsymbol{e}(\boldsymbol{R}^{\boldsymbol{A}}(\boldsymbol{X}_1,\!\boldsymbol{X}_2)\boldsymbol{X}_3,\boldsymbol{X}_4) -\boldsymbol{e}(\boldsymbol{R}^{\boldsymbol{A}}(\boldsymbol{X}_1,\boldsymbol{X}_2)\boldsymbol{X}_4,\boldsymbol{X}_3), \nonumber
\end{align}
%-----------------------------
with
%-----------------------------
\begin{equation}
\begin{aligned}%\label{LDefTT}
&\mathbf{L}^{\boldsymbol{A}}(\boldsymbol{e})(\boldsymbol{X}_{1},\boldsymbol{X}_{2},\boldsymbol{X}_{3},\boldsymbol{X}_{4})= 
 \left(\nabla^{\boldsymbol{A}}_{\boldsymbol{X}_{1}}\nabla^{\boldsymbol{A}}_{\boldsymbol{X}_{3}}\boldsymbol{e}\right)(\boldsymbol{X}_{2},\boldsymbol{X}_{4})+\! \left(\nabla^{\boldsymbol{A}}_{\boldsymbol{X}_{2}}\nabla^{\boldsymbol{A}}_{\boldsymbol{X}_{4}}\boldsymbol{e}\right)(\boldsymbol{X}_{1},\boldsymbol{X}_{3})  \\
&    -\! \left(\nabla^{\boldsymbol{A}}_{\boldsymbol{X}_{1}}\nabla^{\boldsymbol{A}}_{\boldsymbol{X}_{4}}\boldsymbol{e}\right)(\boldsymbol{X}_{2},\boldsymbol{X}_{3}) -\! \left(\nabla^{\boldsymbol{A}}_{\boldsymbol{X}_{2}}\nabla^{\boldsymbol{A}}_{\boldsymbol{X}_{3}}\boldsymbol{e}\right)(\boldsymbol{X}_{1},\boldsymbol{X}_{4})
-\! \left(\nabla^{\boldsymbol{A}}_{\nabla^{\boldsymbol{A}}_{\boldsymbol{X}_{1}}\boldsymbol{X}_{3}}\boldsymbol{e}\right)(\boldsymbol{X}_{2},\boldsymbol{X}_{4}) \\
&  - \left(\nabla^{\boldsymbol{A}}_{\nabla^{\boldsymbol{A}}_{\boldsymbol{X}_{2}}\boldsymbol{X}_{4}}\boldsymbol{e}\right)(\boldsymbol{X}_{1},\boldsymbol{X}_{3}) +\! \left(\nabla^{\boldsymbol{A}}_{\nabla^{\boldsymbol{A}}_{\boldsymbol{X}_{1}}\boldsymbol{X}_{4}}\boldsymbol{e}\right)(\boldsymbol{X}_{2},\boldsymbol{X}_{3}) +\! \left(\nabla^{\boldsymbol{A}}_{\nabla^{\boldsymbol{A}}_{\boldsymbol{X}_{2}}\boldsymbol{X}_{3}}\boldsymbol{e}\right)(\boldsymbol{X}_{1},\boldsymbol{X}_{4}), \nonumber
\end{aligned}
\end{equation}
%-----------------------------
where $\nabla^{\boldsymbol{A}}\boldsymbol{T}$ for $({}^{0}_{k})$-tensor $\boldsymbol{T}$ is defined as  
%-----------------------------
\begin{equation}
\begin{aligned}
\left(\nabla^{\boldsymbol{A}}_{\boldsymbol{X}_{0}}\boldsymbol{T}\right)(\boldsymbol{X}_{1},\dots,\boldsymbol{X}_{k}) =& 		
	\boldsymbol{X}_{0}\left(\boldsymbol{T}(\boldsymbol{X}_{1},\dots,\boldsymbol{X}_{k})\right) \\
&	- \sum_{i=1}^{k}\boldsymbol{T}(\boldsymbol{X}_{1},\dots,\nabla^{\boldsymbol{A}}_{\boldsymbol{X}_{0}}\boldsymbol{X}_{i},\dots,\boldsymbol{X}_{k}). \nonumber
\end{aligned}
\end{equation}
%-----------------------------
Note that $\mathbf{r}^{\boldsymbol{A}}(\boldsymbol{e})$ inherits the symmetries of the Riemannian curvature. If $(\mathcal{B},\boldsymbol{G})$ has constant curvature $\mathrm{k}$, by using (\ref{com_con_sct}), one obtains the operator $\mathrm{D}_{1}:\Gamma(S^{2}T^{\ast}\mathcal{B})\rightarrow \Gamma(S^{2}(\Lambda^{2}T^{\ast}\mathcal{B}))$, $\mathrm{D}_{1}:=2\mathbf{r}^{\boldsymbol{G}}$, which can be written as
%-----------------------------
\begin{equation}\label{Lin_Comp_thm}
\begin{aligned}
	(\mathrm{D}_{1}\boldsymbol{e})(&\boldsymbol{X}_{1},\boldsymbol{X}_{2},\boldsymbol{X}_{3},\boldsymbol{X}_{4})
	=\mathbf{L}^{\boldsymbol{G}}(\boldsymbol{e})(\boldsymbol{X}_{1},\boldsymbol{X}_{2},\boldsymbol{X}_{3},\boldsymbol{X}_{4})  \\  
	& + \mathrm{k}\Big\{\boldsymbol{G}(\boldsymbol{X}_{2},\boldsymbol{X}_{3})\boldsymbol{e}(\boldsymbol{X}_{1},\boldsymbol{X}_{4})- \boldsymbol{G}(\!\boldsymbol{X}_{1},\boldsymbol{X}_{3})\boldsymbol{e}(\boldsymbol{X}_{2},\boldsymbol{X}_{4})  \\
	& -\boldsymbol{G}(\boldsymbol{X}_{2},\boldsymbol{X}_{4})\boldsymbol{e}(\boldsymbol{X}_{1},\boldsymbol{X}_{3}) + \boldsymbol{G}(\boldsymbol{X}_{1},\boldsymbol{X}_{4})\boldsymbol{e}(\boldsymbol{X}_{2},\boldsymbol{X}_{3})\Big\}. 
\end{aligned}
\end{equation}
%-----------------------------
One can show that $\mathrm{D}_{1}\circ\mathrm{D}_{0}=0$. The Calabi complex for a $2$-manifold $\mathcal{B}$ reads
%-----------------------------
\begin{equation}
\begin{gathered}\label{2DCaComp}
\scalebox{1}{\xymatrix@C=3ex{0 \ar[r] &\mathfrak{X}(\mathcal{B}) \ar[r]^-{\mathrm{D}_{0}} &\Gamma(S^{2}T^{\ast}\mathcal{B})\ar[r]^-{\mathrm{D}_{1}} &\Gamma(S^{2}(\Lambda^{2}T^{\ast}\mathcal{B}))  \ar[r] &0 . } } 
\end{gathered}
\end{equation}
%-----------------------------

The last non-trivial operator of the Calabi complex for $3$-manifolds is defined as follows. Let $\Gamma(V^{5}T^{\ast}\mathcal{B})$ be the space of $({}^{0}_{5})$-tensors such that $\boldsymbol{h}\in\Gamma(V^{5}T^{\ast}\mathcal{B})$ admits the following symmetries.
%-----------------------------
\begin{equation}
\begin{aligned}
h_{I_1I_2I_3I_4I_5}&=-h_{I_2I_1I_3I_4I_5}=-h_{I_1I_3I_2I_4I_5}, \\   \nonumber
h_{I_{1}I_{2}I_{3}I_{4}I_{5}}&+h_{I_{1}I_{3}I_{4}I_{2}I_{5}}+h_{I_{1}I_{4}I_{2}I_{3}I_{5}}=0, \\
h_{I_{1}I_{2}I_{3}I_{4}I_{5}}&=-h_{I_{1}I_{3}I_{2}I_{4}I_{5}}=-h_{I_{1}I_{2}I_{3}I_{5}I_{4}},
\end{aligned}
\end{equation}
%-----------------------------
i.e. $\boldsymbol{h}$ is anti-symmetric in the first three entries and has the symmetries of the Riemannian curvature in the last four entries. For $n=3$, $\boldsymbol{h}$ has $3$ independent components that can be represented by $h_{12323}$, $h_{21313}$, and $h_{31212}$. The operator $\mathrm{D}_{2}:\Gamma(S^{2}(\Lambda^{2}T^{\ast}\mathcal{B}))\rightarrow\Gamma(V^{5}T^{\ast}\mathcal{B})$ is defined by
%-----------------------------
\begin{equation}
\begin{aligned}
(\mathrm{D}_{2}\boldsymbol{s})(&\boldsymbol{X}_1,\dots,\boldsymbol{X}_5) =  \left(\nabla_{\boldsymbol{X}_1}\boldsymbol{s}\right)(\boldsymbol{X}_2,\boldsymbol{X}_3,\boldsymbol{X}_4,\boldsymbol{X}_5) \\
&+ \left(\nabla_{\boldsymbol{X}_2}\boldsymbol{s}\right)(\boldsymbol{X}_3,\boldsymbol{X}_1,\boldsymbol{X}_4,\boldsymbol{X}_5) + \left(\nabla_{\boldsymbol{X}_3}\boldsymbol{s}\right)(\boldsymbol{X}_1,\boldsymbol{X}_2,\boldsymbol{X}_4,\boldsymbol{X}_5). \nonumber
\end{aligned}
\end{equation}
%-----------------------------
By using $\mathrm{D}_{2}$, the second Bianchi identity for the Riemannian curvature $\boldsymbol{\mathcal{R}}$ can be expressed as $\mathrm{D}_{2}(\boldsymbol{\mathcal{R}})=0$. We have $\mathrm{D}_{2}\circ\mathrm{D}_{1}=0$. Thus, the Calabi complex on a $3$-manifold $\mathcal{B}$ is written as
%-----------------------------
\begin{equation}
\begin{gathered}\label{3DCaComp}
\scalebox{1}{\xymatrix@C=3ex{0 \ar[r] &\mathfrak{X}(\mathcal{B}) \ar[r]^-{\mathrm{D}_{0}} &\Gamma(S^{2}T^{\ast}\mathcal{B}) \ar[r]^-{\mathrm{D}_{1}} &\Gamma(S^{2}(\Lambda^{2}T^{\ast}\mathcal{B})) \ar[r]^-{\mathrm{D}_{2}} & \Gamma(V^{5}T^{\ast}\mathcal{B})  \ar[r] &0 . } } 
\end{gathered}
\end{equation}
%-----------------------------
\citet{Calabi1961} showed that there is a systematic way for constructing operators $\mathrm{D}_{i}$, $2\leq i\leq n-1$, for $n$-manifolds. Let $\mathrm{H}^{k}_{C}(\mathcal{B}):=\ker\mathrm{D}_{k}/\mathrm{im}\,\mathrm{D}_{k-1}$ be the $k$-th cohomology group of the Calabi complex. Calabi also showed that the Calabi complex induces a fine resolution of the sheaf of germs of Killing vector fields. Thus, the dimension of $\text{\textTheta}(\mathcal{B})$ determines the dimension of $\mathrm{H}^{k}_{C}(\mathcal{B})$. In particular, if an $n$-manifold $\mathcal{B}\subset\mathbb{R}^{n}$ has finite-dimensional de Rham cohomology groups, one can write  
%-----------------------------
\begin{equation}\label{CalCohom}
\dim \mathrm{H}^{k}_{C}(\mathcal{B})=\frac{n(n+1)}{2}\dim \mathrm{H}^{k}_{dR}(\mathcal{B})=\frac{n(n+1)}{2}b_{k}(\mathcal{B}). 
\end{equation}
%-----------------------------
Next, we separately consider $2$- and $3$-submanifolds of the Euclidean space.

%-----------------------------
%-----------------------------
\subsubsection{The linear elasticity complex for $3$-manifolds}

Let $\mathcal{B}\subset\mathbb{R}^{3}$ be an open subset and let $\boldsymbol{G}$ and $\{X^{I}\}$ be the Euclidean metric and the Cartesian coordinates of $\mathbb{R}^{3}$, respectively. Consider the following operator  
%-----------------------------
\begin{equation}
\mathbf{grad}^{s}:\mathfrak{X}(\mathcal{B})\rightarrow\Gamma(S^{2}T\mathcal{B}),~\left(\mathbf{grad}^{s}\boldsymbol{Y}\right)^{IJ}=\frac{1}{2}\left( Y^{I}{}_{,J} + Y^{J}{}_{,I}\right). \nonumber
\end{equation}
%-----------------------------
It is straightforward to show that $\mathbf{curl}\circ\mathbf{curl}\circ\mathbf{grad}^{s}=0$. If $\boldsymbol{T}\in\Gamma(S^{2}T\mathcal{B})$, then $\mathbf{curl}\circ\mathbf{curl}\,\boldsymbol{T}$ is symmetric as well. Therefore, one obtains the following operator
%-----------------------------
\begin{equation}
\mathbf{curl}\circ\mathbf{curl}:\Gamma(S^{2}T\mathcal{B})\rightarrow\Gamma(S^{2}T\mathcal{B}),~\left(\mathbf{curl}\circ\mathbf{curl}\,\boldsymbol{T}\right)^{IJ}=\varepsilon_{IKL}\varepsilon_{JMN}T^{LN}{}_{,KM}. \nonumber
\end{equation}
%-----------------------------
We have $\mathbf{div}\circ\mathbf{curl}\circ\mathbf{curl}=0$. Let $\boldsymbol{\iota}_{0}:\mathfrak{X}(\mathcal{B})\rightarrow\mathfrak{X}(\mathcal{B})$ be the identity map. The global orthonormal coordinate system $\{X^{I}\}$ allows one to define the following three isomorphisms: 
%-----------------------------
\begin{equation}
\boldsymbol{\iota}_{1}:\Gamma(S^{2}T\mathcal{B})\rightarrow\Gamma(S^{2}T^{\ast}\mathcal{B}),~\left(\boldsymbol{\iota}_{1}(\boldsymbol{T})\right)_{IJ}=T^{IJ}, \nonumber
\end{equation}
%-----------------------------
the isomorphism $\boldsymbol{\iota}_{2}:\Gamma(S^{2}T\mathcal{B})\rightarrow\Gamma(S^{2}(\Lambda^{2}T^{\ast}\mathcal{B}))$ defined by 
%-----------------------------
\begin{equation}
\begin{aligned}
&\left(\boldsymbol{\iota}_{2}(\boldsymbol{T})\right)_{2323}=T^{11},~\left(\boldsymbol{\iota}_{2}(\boldsymbol{T})\right)_{3123}=T^{12},~\left(\boldsymbol{\iota}_{2}(\boldsymbol{T})\right)_{1223}=T^{13}, \\&\left(\boldsymbol{\iota}_{2}(\boldsymbol{T})\right)_{1313}=T^{22},~\left(\boldsymbol{\iota}_{2}(\boldsymbol{T})\right)_{2113}=T^{23},~ \left(\boldsymbol{\iota}_{2}(\boldsymbol{T})\right)_{1212} =T^{33}, \nonumber
\end{aligned}
\end{equation}
%-----------------------------   
and $\boldsymbol{\iota}_{3}:\mathfrak{X}(\mathcal{B})\rightarrow\Gamma(V^{5}T^{\ast}\mathcal{B}))$ given by
%-----------------------------
\begin{equation}
\left(\boldsymbol{\iota}_{3}(\boldsymbol{Y})\right)_{12323}=Y^{1},~~~\left(\boldsymbol{\iota}_{3}(\boldsymbol{Y})\right)_{21313}=Y^{2},~~~\left(\boldsymbol{\iota}_{3}(\boldsymbol{Y})\right)_{31212}=Y^{3}. \nonumber
\end{equation}
%-----------------------------
Simple calculations show that
%-----------------------------
\begin{equation}
 \boldsymbol{\iota}_{1}\circ\mathbf{grad}^{s}=\mathrm{D}_{0}\circ\boldsymbol{\iota}_{0},~~
 \boldsymbol{\iota}_{2}\circ\mathbf{curl}\circ\mathbf{curl}=\mathrm{D}_{1}\circ\boldsymbol{\iota}_{1}, ~~
 \boldsymbol{\iota}_{3}\circ\mathbf{div}=\mathrm{D}_{2}\circ\boldsymbol{\iota}_{2}, \nonumber 
\end{equation}
%----------------------------- 
and therefore, the following diagram commutes.
%-----------------------------
\begin{equation}
\begin{gathered}\label{3DElasCal}
\scalebox{1}{\xymatrix@C=3ex{0 \ar[r] &\mathfrak{X}(\mathcal{B}) \ar[r]^-{\mathbf{grad}^{s}} \ar[d]^{\boldsymbol{\iota}_{0}} &\Gamma(S^{2}T\mathcal{B}) \ar[r]^-{\mathbf{curl}\circ\mathbf{curl}} \ar[d]^{\boldsymbol{\iota}_{1}} &\Gamma(S^{2}T\mathcal{B}) \ar[r]^-{\mathbf{div}} \ar[d]^{\boldsymbol{\iota}_{2}} &\mathfrak{X}(\mathcal{B}) \ar[r] \ar[d]^{\boldsymbol{\iota}_{3}} &0  \\
0 \ar[r] &\mathfrak{X}(\mathcal{B}) \ar[r]^-{\mathrm{D}_{0}} &\Gamma(S^{2}T^{\ast}\mathcal{B}) \ar[r]^-{\mathrm{D}_{1}} &\Gamma(S^{2}(\Lambda^{2}T^{\ast}\mathcal{B})) \ar[r]^-{\mathrm{D}_{2}} & \Gamma(V^{5}T^{\ast}\mathcal{B})  \ar[r] &0 } }  
\end{gathered}
\end{equation}
%-----------------------------
The first row of (\ref{3DElasCal}) is the linear elasticity complex. Therefore, we observe that useful properties of this complex follow from those of the Calabi complex. In particular, (\ref{CalCohom}) implies that the dimensions of the cohomology groups $\mathrm{H}^{k}_{E3}(\mathcal{B})$ of the linear elasticity complex are given by 
%-----------------------------
\begin{equation}\label{E3CalCohom}
\begin{aligned}
\dim \mathrm{H}^{1}_{E3}(\mathcal{B})&:= \dim\left(\ker \mathbf{curl}\circ\mathbf{curl}/\mathrm{im}\,\mathbf{grad}^{s}\right)=6 b_{1}(\mathcal{B}), \nonumber \\ 
\dim \mathrm{H}^{2}_{E3}(\mathcal{B})&:= \dim\left(\ker \mathbf{div}/\mathrm{im}\,\mathbf{curl}\circ\mathbf{curl}\right)=6 b_{2}(\mathcal{B}).
\end{aligned}
\end{equation}
%-----------------------------
We should mention that it is possible to calculate the cohomology groups of the linear elasticity complex without explicitly using its relation with the Calabi complex \citep{HacklZastrow1988}. Note also that the Calabi complex is more general than the linear elasticity complex in the sense that the Calabi complex is valid on any Riemannian manifold with constant curvature. However, the linear elasticity complex is only valid on flat manifolds that admit a global orthonormal coordinate system. \citet[Proposition 2.8]{Yavari2013} showed that for $\boldsymbol{T}\in\Gamma(S^{2}T\mathcal{B})$, there exists $\boldsymbol{Y}\in\mathfrak{X}(\mathcal{B})$ such that $\boldsymbol{T}=\mathbf{grad}^{s}\boldsymbol{Y}$, if and only if
%-----------------------------
\begin{equation}\label{E3ExactCom}
\begin{aligned}
&\mathbf{curl}\circ\mathbf{curl}\,\boldsymbol{T}=0, \\ &
\int_{\ell} \left[ T^{IJ} - X^{k}(T^{IJ}{}_{,K} - T^{JK}{}_{,I})   \right]dX^{J}=0,  \\ &\int_{\ell} \left( T^{IK}{}_{,J}-T^{JK}{}_{,I} \right)dX^{K}=0, ~\forall\ell\subset\mathcal{B}.
\end{aligned}
\end{equation}
%-----------------------------
Let $\boldsymbol{N}_{\!\mathcal{C}}$ be the unit outward normal vector field of an arbitrary closed surface $\mathcal{C}\subset \mathcal{B}$. \citet{Gurtin1963} showed that the necessary and sufficient conditions for the existence of $\mathbf{curl}\circ\mathbf{curl}$-potentials for $\boldsymbol{T}$ are
%-----------------------------
\begin{equation}\label{ExactcurlcurlP}
\mathbf{div}\,\boldsymbol{T}=0, \!\int_{\mathcal{C}}\langle\boldsymbol{T},\boldsymbol{N}_{\!\mathcal{C}}\rangle dA=0, \!\int_{\mathcal{C}}\varepsilon_{KIJ}X^{I} T^{JL} (\boldsymbol{N}_{\!\mathcal{C}})^{L} dA=0, ~\forall \mathcal{C} \subset\mathcal{B}. 
\end{equation}
%----------------------------- 

%----------------------------- 
%----------------------------- 
\subsubsection{The linear elasticity complex for $2$-manifolds}

Next, suppose $\mathcal{B}\subset\mathbb{R}^{2}$ is a $2$-manifold and let $\{X^{I}\}$ be the Cartesian coordinates of $\mathbb{R}^{2}$. Let 
%-----------------------------
\begin{equation}
\mathrm{D}_{c}:\Gamma(S^{2}T\mathcal{B})\rightarrow C^{\infty}(\mathcal{B}),~\mathrm{D}_{c}\boldsymbol{T}=T^{11}{}_{,22}-2T^{12}{}_{,12} +T^{22}{}_{,11}. \nonumber
\end{equation}
%-----------------------------
Then, we have $\mathrm{D}_{c}\circ\mathbf{grad}^{s}=0$. Also consider isomorphisms $\text{\textgamma}_{0}$, $\text{\textgamma}_{1}$, and $\text{\textgamma}_{2}$ that are defined as follows: $\text{\textgamma}_{0}:\mathfrak{X}(\mathcal{B})\rightarrow\mathfrak{X}(\mathcal{B})$ and $\text{\textgamma}_{1}:\Gamma(S^{2}T\mathcal{B})\rightarrow\Gamma(S^{2}T^{\ast}\mathcal{B})$
are defined similarly to $\boldsymbol{\iota}_0$ and $\boldsymbol{\iota}_1$ for $3$-manifolds and $\text{\textgamma}_{2}:C^{\infty}(\mathcal{B})\rightarrow\Gamma(S^{2}(\Lambda^{2}T^{\ast}\mathcal{B}))$, $\left(\text{\textgamma}_{2}(f)\right)_{1212}=f$. Using these operators, one obtains the following diagram.   
%-----------------------------
\begin{equation}
\begin{gathered}\label{2DElasCal}
\scalebox{1}{\xymatrix@C=3ex{0 \ar[r] &\mathfrak{X}(\mathcal{B}) \ar[r]^-{\mathbf{grad}^{s}} \ar[d]^{\text{\textgamma}_{0}} &\Gamma(S^{2}T\mathcal{B}) \ar[r]^-{\mathrm{D}_{c}} \ar[d]^{\text{\textgamma}_{1}} &C^{\infty}(\mathcal{B}) \ar[r] \ar[d]^{\text{\textgamma}_{2}}  &0  \\
0 \ar[r] &\mathfrak{X}(\mathcal{B}) \ar[r]^-{\mathrm{D}_{0}} &\Gamma(S^{2}T^{\ast}\mathcal{B}) \ar[r]^-{\mathrm{D}_{1}} &\Gamma(S^{2}(\Lambda^{2}T^{\ast}\mathcal{B}))   \ar[r] &0 } } 
\end{gathered}
\end{equation}
%-----------------------------
Therefore, (\ref{CalCohom}) implies that the dimension of the cohomology group $\mathrm{H}^{1}_{E2}(\mathcal{B}):= \ker \mathrm{D}_{c}/\mathrm{im}\,\mathbf{grad}^{s}$ is $3 b_{1}(\mathcal{B})$. Moreover, the necessary and sufficient conditions for the existence of potentials induced by $\mathbf{grad}^{s}$ for $\boldsymbol{T}\in\Gamma(S^{2}T\mathcal{B})$ is $\mathrm{D}_{c}\boldsymbol{T}=0$, together with the integral conditions in (\ref{E3ExactCom}). For $2$-manifolds, it is also possible to write the complex
%-----------------------------
\begin{equation}
\begin{gathered}\label{2DDivElasCal}
\scalebox{1}{\xymatrix@C=3ex{0 \ar[r] &C^{\infty}(\mathcal{B}) \ar[r]^-{\mathrm{D}_{s}} &\Gamma(S^{2}T\mathcal{B}) \ar[r]^-{\mathbf{div}} &\mathfrak{X}(\mathcal{B}) \ar[r]  &0, } }  
\end{gathered}
\end{equation}
%-----------------------------
where $(\mathrm{D}_{s}f)^{11}=f_{,22}$, $(\mathrm{D}_{s}f)^{12}=-f_{,12}$, and $(\mathrm{D}_{s}f)^{22}=f_{,11}$. The kernel of $\mathrm{D}_{s}$ is $3$-dimensional, which suggests that the dimension of $\mathrm{H}^{1}_{E2'}(\mathcal{B}):= \ker \mathbf{div}/\mathrm{im}\,\mathrm{D}_{s}$, is $3 b_{1}(\mathcal{B})$. By replacing $\mathcal{C}$ with arbitrary closed curves $\ell$ in (\ref{ExactcurlcurlP}), one obtains the necessary and sufficient conditions for the existence of $\mathrm{D}_{s}$-potentials.

%-----------------------------
%-----------------------------
\section{Some Applications in Continuum Mechanics}

Let $\mathcal{B}\subset\mathbb{R}^{n}$, $n=2,3$, be a smooth $n$-manifold. Note that $\mathcal{B}$ can be unbounded as well. For $3$-manifolds, the linear elasticity complex (\ref{3DElasCal}) describes both the kinematics and the kinetics of deformations in the following sense \citep{Kroner1959}: If one considers $\mathfrak{X}(\mathcal{B})$ as the space of displacements, then $\mathbf{grad}^{s}$ associates linear strains to displacements, $\Gamma(S^{2}T\mathcal{B})$ is the space of linear strains, and $\mathbf{curl}\circ\mathbf{curl}$ expresses the compatibility equations for the linear strain. On the other hand, one can consider $\Gamma(S^{2}T\mathcal{B})$ as the space of Beltrami stress functions and consequently, $\mathbf{curl}\circ\mathbf{curl}$ associates symmetric Cauchy stress tensors to Beltrami stress functions, and $\mathbf{div}$ expresses the equilibrium equations. We observed that for $2$-manifolds, the kinematics and the kinetics of deformation are described by two separate complexes: The former is addressed by the complex (\ref{2DElasCal}) and the latter by the complex (\ref{2DDivElasCal}). 

Let a smooth embedding $\varphi:\mathcal{B}\rightarrow\mathcal{S}=\mathbb{R}^{3}$ be a motion of $\mathcal{B}$ in $\mathcal{S}$. Let $\boldsymbol{C}:=\varphi^{\ast}\boldsymbol{g}\in\Gamma(S^{2}T^{\ast}\mathcal{B})$, and $\boldsymbol{F}:=T\varphi$ be the Green deformation tensor and the deformation gradient of $\varphi$, respectively. Also suppose $\boldsymbol{\sigma}\in\Gamma(\text{\large{$\otimes$}}^{2}T\varphi(\mathcal{B}))$, $\boldsymbol{P}\in\Gamma(T\varphi(\mathcal{B})\otimes T\mathcal{B})$, and $\boldsymbol{S}\in\Gamma(\text{\large{$\otimes$}}^{2}T\mathcal{B})$ are the Cauchy, the first, and the second Piola-Kirchhoff stress tensors, respectively. If $\boldsymbol{\sigma}$ is symmetric, then the last two operators of the linear elasticity complex on $(\varphi(\mathcal{B}),\boldsymbol{g})$ address existence of Beltrami stress functions and the equilibrium equations for $\boldsymbol{\sigma}$. The first operator in this complex does not have any apparent physical interpretation. If $\boldsymbol{\sigma}$ is non-symmetric, then the last two operators of the $\boldsymbol{\mathsf{gcd}}$ complex on $(\varphi(\mathcal{B}),\boldsymbol{g})$ describe the kinetics of $\varphi$: The operator $\mathbf{curl}^{\mathsf{T}}$ associates stress functions induced by $\mathbf{curl}^{\mathsf{T}}$ to $\boldsymbol{\sigma}$ and $\mathbf{div}$ is related to the equilibrium equations. Similar conclusions also hold for $\boldsymbol{S}$ if one considers the linear elasticity complex and the $\boldsymbol{\mathsf{gcd}}$ complex on the flat manifold $(\mathcal{B},\boldsymbol{C})$. 

On the other hand, by using $\boldsymbol{P}$, one can write a complex that describes both the kinematics and the kinetics of motion. Let $\boldsymbol{U}\in\Gamma(T\varphi(\mathcal{B}))$ be the displacement field defined as $\boldsymbol{U}(X)=\varphi(X)-X\in T_{\varphi(X)}\mathcal{S}$, $\forall X\in\mathcal{B}$.\footnote{Displacement fields are usually assumed to be vector fields on $\mathcal{B}$. The choice of $\Gamma(T\varphi(\mathcal{B}))$ instead of $\mathfrak{X}(\mathcal{B})$ is equivalent to applying the shifter $T_{X}\mathcal{B}\rightarrow T_{\varphi(X)}\mathcal{S}$ to elements of $\mathfrak{X}(\mathcal{B})$, see \citep[Box 3.1]{MarsdenHughes1994}.} Then, $\mathbf{Grad}\,\boldsymbol{U}$ is the displacement gradient and $\mathbf{Curl}^{\mathsf{T}}$ expresses the compatibility of the displacement gradient. On the other hand, we can assume that $\mathbf{Curl}^{\mathsf{T}}$ associates stress functions induced by $\mathbf{Curl}^{\mathsf{T}}$ to the first Piola-Kirchhoff stress tensor, and that $\mathbf{Div}$ expresses the equilibrium equations. Hence, the $\boldsymbol{\mathsf{GCD}}$ complex is the nonlinear analogue of the linear elasticity complex in the sense that both contain the kinematics and the kinetics of motion simultaneously. Note that the linear elasticity complex is not the linearization of the $\boldsymbol{\mathsf{GCD}}$ complex. In particular, the operator $\mathbf{curl}\circ\mathbf{curl}$ is obtained by linearizing the curvature operator, which is related to the compatibility equations in terms of $\boldsymbol{C}$ and not the displacement gradient.

In the following, we study the applications of the above complexes to the nonlinear compatibility equations and the existence of stress functions in more detail. Classically, the linear and nonlinear compatibility equations are written for flat ambient spaces. We study these equations on ambient spaces with constant curvatures as well.

%-----------------------------
%-----------------------------
\subsection{Compatibility Equations}

We study the nonlinear compatibility equations for the cases $\dim\mathcal{B}=\dim\mathcal{S}$, and $\dim\mathcal{B}<\dim\mathcal{S}$ (shells), separately. It is well-known that compatibility equations depend on the topological properties of bodies, see \citet{Yavari2013} and references therein for more details. More specifically, both linear and nonlinear compatibility equations are closely related to $b_{1}(\mathcal{B})$. The nonlinear compatibility equations in terms of the displacement gradient (or equivalently $\boldsymbol{F}$) directly follow from the complexes we introduced earlier for second-order tensors. 

%-----------------------------
%-----------------------------
\subsubsection{Bodies with the same dimensions as the ambient space}

Suppose $\dim\mathcal{B}=\dim\mathcal{S}$. Since motion $\varphi:\mathcal{B}\rightarrow\mathcal{S}$ is an embedding, it is easy to observe that the Green deformation tensor $\boldsymbol{C}=\varphi^{\ast}\boldsymbol{g}$ is a Riemannian metric on $\mathcal{B}$. The mapping $\varphi$ is an isometry between $(\mathcal{B},\boldsymbol{C})$ and $(\varphi(\mathcal{B}),\boldsymbol{g})$. Thus, the compatibility problem in terms of $\boldsymbol{C}$ reads: Given a metric $\boldsymbol{C}$ on $\mathcal{B}$, is there any isometry between $(\mathcal{B},\boldsymbol{C})$ and an open subset of $\mathcal{S}$? Note that a priori we do not know which part of $\mathcal{S}$ would be occupied by $\mathcal{B}$. This suggests that a useful compatibility equation should be written only on $\mathcal{B}$. Let $\boldsymbol{R}^{\boldsymbol{g}}$ and $\boldsymbol{\mathcal{R}}^{\boldsymbol{g}}$ be the curvature and the Riemannian curvature of $(\mathcal{S},\boldsymbol{g})$ that are induced by the Levi-Civita connection $\nabla^{\boldsymbol{g}}$. Let $\boldsymbol{X}_{\!1},\dots,\boldsymbol{X}_{\!4}\in\mathfrak{X}(\mathcal{B})$. By using the pull-back $\varphi^{\ast}$ and the push-forward $\varphi_{\ast}$, one concludes that the linear connection $\nabla^{\boldsymbol{g}}$ on $T\mathcal{S}$ induces a linear connection $\varphi^{\ast}\nabla^{\boldsymbol{g}}$ on $T\mathcal{B}$ given by $(\varphi^{\ast}\nabla^{\boldsymbol{g}})_{\boldsymbol{X}_{1}}\!\boldsymbol{X}_{2}=\varphi^{\ast}\big(\nabla^{\boldsymbol{g}}_{\!\varphi_{\ast}\boldsymbol{X}_{1}}\!\varphi_{\ast}\boldsymbol{X}_{2}\big)$. The definition of the Levi-Civita connection $\nabla^{\boldsymbol{g}}$ implies that  
%-----------------------------
\begin{align}
 \boldsymbol{C}(&\boldsymbol{X}_3,(\varphi^{\ast}\nabla^{\boldsymbol{g}})_{\!\boldsymbol{X}_1}\boldsymbol{X}_2)
= \boldsymbol{g}(\varphi_{\ast}\boldsymbol{X}_3,\nabla^{\boldsymbol{g}}_{\varphi_{\ast}\boldsymbol{X}_1}\,\varphi_{\ast}\boldsymbol{X}_2) \nonumber\\ 
& =\frac{1}{2}\Big\{\boldsymbol{X}_2\left(\boldsymbol{C}(\boldsymbol{X}_1,\boldsymbol{X}_3)\right) + \boldsymbol{X}_1\left(\boldsymbol{C}(\boldsymbol{X}_3,\boldsymbol{X}_2)\right)-\boldsymbol{X}_3\left(\boldsymbol{C}(\boldsymbol{X}_1,\boldsymbol{X}_2)\right) \nonumber \\
& -\boldsymbol{C}\left([\boldsymbol{X}_2,\boldsymbol{X}_3],\boldsymbol{X}_1\right)-\boldsymbol{C}\left([\boldsymbol{X}_1,\boldsymbol{X}_3],\boldsymbol{X}_2\right)-\boldsymbol{C}\left([\boldsymbol{X}_2,\boldsymbol{X}_1],\boldsymbol{X}_3\right) \Big\}, \nonumber
\end{align}
%-----------------------------
and therefore, $\varphi^{\ast}\nabla^{\boldsymbol{g}}$ coincides with the Levi-Civita connection $\nabla^{\boldsymbol{C}}$ on $(\mathcal{B},\boldsymbol{C})$. Since
%-----------------------------
\begin{align}
	 (\!\varphi^{\ast}\boldsymbol{R}^{\boldsymbol{g}})(\boldsymbol{X}_1,\boldsymbol{X}_2)\boldsymbol{X}_3 &=\varphi^{\ast}\left(\boldsymbol{R}^{\boldsymbol{g}}(\varphi_{\ast}\boldsymbol{X}_1,\varphi_{\ast}\boldsymbol{X}_2)\varphi_{\ast}\boldsymbol{X}_3\right) \nonumber\\
	&= \nabla^{\boldsymbol{C}}_{\boldsymbol{X}_1}\nabla^{\boldsymbol{C}}_{\boldsymbol{X}_2}\boldsymbol{X}_3-\nabla^{\boldsymbol{C}}_{\boldsymbol{X}_2}\nabla^{\boldsymbol{C}}_{\boldsymbol{X}_1}\boldsymbol{X}_3-\nabla^{\boldsymbol{C}}_{[\boldsymbol{X}_1,\boldsymbol{X}_2]}\boldsymbol{X}_3, \nonumber
\end{align}
%-----------------------------
we also conclude that $\varphi^{\ast}\boldsymbol{R}^{\boldsymbol{g}}$ is the curvature $\boldsymbol{R}^{\boldsymbol{C}}$ of $(\mathcal{B},\boldsymbol{C})$ induced by $\nabla^{\boldsymbol{C}}$. Therefore, if $\varphi:\mathcal{B}\rightarrow\mathcal{S}$ is an isometry between $(\mathcal{B},\boldsymbol{C})$ and $(\varphi(\mathcal{B}),\boldsymbol{g})$, then we must have 
%-----------------------------
\begin{equation}\label{Curv_condC}
\boldsymbol{\mathcal{R}}^{\boldsymbol{C}}(\boldsymbol{X}_{\!1},\boldsymbol{X}_{\!2},\boldsymbol{X}_{\!3},\boldsymbol{X}_{\!4})=\boldsymbol{\mathcal{R}}^{\boldsymbol{g}}(\varphi_{\ast}\boldsymbol{X}_{\!1},\varphi_{\ast}\boldsymbol{X}_{\!2},\varphi_{\ast}\boldsymbol{X}_{\!3},\varphi_{\ast}\boldsymbol{X}_{\!4}), 
\end{equation}
%-----------------------------
where $\boldsymbol{\mathcal{R}}^{\boldsymbol{C}}$ is the Riemannian curvature of $(\mathcal{B},\boldsymbol{C})$. It is hard to check the above condition on arbitrary curved ambient spaces. However, if $\mathcal{S}$ has constant curvature, (\ref{Curv_condC}) admits a simple form. The following theorem states the compatibility equations in terms of $\boldsymbol{C}$ on an ambient space with constant curvature.

\begin{thm}\label{ThmLoc_ConCur} Suppose $\dim\mathcal{B}=\dim\mathcal{S}$, and $(\mathcal{S},\boldsymbol{g})$ has constant curvature $\widehat{\mathrm{k}}$. If $\boldsymbol{C}$ is the Green deformation tensor of a motion $\varphi:\mathcal{B}\rightarrow\mathcal{S}$, then $(\mathcal{B},\boldsymbol{C})$ has constant curvature $\widehat{\mathrm{k}}$ as well, i.e. $\boldsymbol{C}$ satisfies
%-----------------------------
\begin{equation}\label{com_con_sc}
\boldsymbol{R}^{\boldsymbol{C}}(\boldsymbol{X}_{\!1},\boldsymbol{X}_{\!2})\boldsymbol{X}_{\!3}=\widehat{\mathrm{k}}\boldsymbol{C}(\boldsymbol{X}_{\!3},\boldsymbol{X}_{\!2})\boldsymbol{X}_{\!1} - \widehat{\mathrm{k}}\boldsymbol{C}(\boldsymbol{X}_{\!3},\boldsymbol{X}_{\!1})\boldsymbol{X}_{\!2}.
\end{equation}
%-----------------------------  
Conversely, if $\boldsymbol{C}$ satisfies (\ref{com_con_sc}), then for each $X\in\mathcal{B}$, there is a neighborhood $\mathcal{U}_{X}\subset\mathcal{B}$ of $X$ and a motion $\varphi_{X}:\mathcal{U}_{X}\rightarrow\mathcal{S}$, with $\boldsymbol{C}|_{\mathcal{U}_{X}}$ being its Green deformation tensor. Motion $\varphi_{X}$ is unique up to isometries of $\mathcal{S}$.    
\end{thm}

\begin{proof} If $\boldsymbol{C}=\varphi^{\ast}\boldsymbol{g}$, then by using $\boldsymbol{R}^{\boldsymbol{C}}=\varphi^{\ast}\boldsymbol{R}^{\boldsymbol{g}}$, and (\ref{com_con_sct}), one obtains (\ref{com_con_sc}). Conversely, consider arbitrary points $X\in\mathcal{B}$ and $x\in\mathcal{S}$ and let $\{\mathbf{E}'_{i}\}$ and $\{\mathbf{e}'_{i}\}$ be arbitrary orthonormal bases for $T_{X}\mathcal{B}$ and $T_{x}\mathcal{S}$, respectively. Choose the isometry $\mathrm{i}:T_{X}\mathcal{B}\rightarrow T_{x}\mathcal{S}$ such that $\mathrm{i}(\mathbf{E}'_{i})=\mathbf{e}'_{i}$. Then, by using a theorem due to Cartan \citep[page 157]{Docarmo1992} and (\ref{com_con_sc}), one can construct an isometry $\varphi_{X}:\mathcal{U}_{X}\rightarrow \varphi_{X}(\mathcal{U}_{X})\subset\mathcal{S}$, in a neighborhood $\mathcal{U}_{X}$ of $X$ such that $T_{X}\varphi_{X}=\mathrm{i}$. This concludes the proof.  \end{proof}

\begin{rem}
Theorem \ref{ThmLoc_ConCur} implies that there are many local isometries between manifolds with the same constant sectional curvatures. Formulating sufficient conditions for the existence of global isometries between arbitrary Riemannian manifolds is a hard problem. \citet{Ambrose1956} derived such a condition by using the parallel translation of Riemannian curvature along curves made up of geodesic segments. In particular, his result implies that (\ref{com_con_sc}) is also a sufficient condition for the existence of a global motion $\varphi:\mathcal{B}\rightarrow\mathcal{S}$, if $\mathcal{B}$ is complete and simply-connected. For the flat case $\mathcal{B}\subset\mathcal{S}=\mathbb{R}^{n}$, \citet{Yavari2013} derived the necessary and sufficient conditions for the compatibility of $\boldsymbol{C}$ when $\mathcal{B}$ is non-simply-connected.
\end{rem}

\begin{rem} 
The symmetries of the Riemannian curvature determine the number of compatibility equations induced by (\ref{com_con_sc}), i.e. the number of independent equations that we obtain by writing (\ref{com_con_sc}) in a local coordinate system. Thus, the number of compatibility equations in terms of $\boldsymbol{C}$ only depends on the dimension of the ambient space and is the same as the number of linear compatibility equations induced by the operator $\mathrm{D}_{1}$ in the Calabi complex.
\end{rem}

Next, suppose $\mathcal{B}\subset\mathcal{S}=\mathbb{R}^{n}$, $n=2,3$, and let $\{X^{I}\}$ and $\{x^{i}\}$ be the Cartesian coordinates of $\mathbb{R}^{n}$. Any smooth mapping $\varphi:\mathcal{B}\rightarrow\mathbb{R}^{n}$ induces a displacement field $\boldsymbol{U}\in\Gamma(T\varphi(\mathcal{B}))$ given by $\boldsymbol{U}(X)=\varphi(X)-X$. One can use the $\boldsymbol{\mathsf{GCD}}$ and the $\boldsymbol{\mathsf{GC}}$ complexes for writing the compatibility equations in terms of the displacement gradient. Note that $\varphi$ is assumed to be specified for writing the above complexes. Let $\boldsymbol{\Upsilon}\in\Omega^{0}(\mathcal{B};\mathbb{R}^{3})$ and $\boldsymbol{\kappa}\in\Omega^{1}(\mathcal{B};\mathbb{R}^{3})$. If $\boldsymbol{\kappa}=\boldsymbol{d}\boldsymbol{\Upsilon}$, then $\boldsymbol{I}^{-1}_{1}(\boldsymbol{\kappa})=\mathbf{Grad}\,\boldsymbol{I}^{-1}_{0}(\boldsymbol{\Upsilon})$, where $\boldsymbol{I}^{-1}_{0}(\boldsymbol{\Upsilon})$ and $\boldsymbol{I}^{-1}_{1}(\boldsymbol{\kappa})$ are two-point tensors over any arbitrary smooth mapping $\varphi$. In particular, by using the linear structure of $\mathbb{R}^{3}$, one can choose $\varphi$ to be $\varphi(X)=X+\boldsymbol{\Upsilon}(X)$. Thus, we obtain the following theorem, cf. Theorems \ref{ExactGradCurlTPTen} and \ref{ExactGradS2pTen2D}.

\begin{thm}\label{CompF23D} Given $\boldsymbol{\kappa}=(\boldsymbol{\kappa}^{1},\dots,\boldsymbol{\kappa}^{n})\in\Omega^{1}(\mathcal{B};\mathbb{R}^{n})$ on a connected $n$-manifold $\mathcal{B}\subset\mathbb{R}^{n}$, there exists a smooth mapping $\varphi:\mathcal{B}\rightarrow\mathbb{R}^{n}$ with displacement gradient $\boldsymbol{I}^{-1}_{1}(\boldsymbol{\kappa})$ (or $\boldsymbol{J}^{-1}_{1}(\boldsymbol{\kappa})$ if $n=2$) if and only if 
%-----------------------------
\begin{equation}\label{CompF23DEq}
\boldsymbol{d}\boldsymbol{\kappa}=0, \text{ and }~~~\int_{\ell}\boldsymbol{\kappa}(\boldsymbol{t}_{\ell})dS=0, ~\forall\ell\subset\mathcal{B}. \nonumber 
\end{equation}
%----------------------------- 
The mapping $\varphi$ is unique up to rigid body translations in $\mathbb{R}^{n}$.   
\end{thm}
%-----------------------------
%-----------------------------
\begin{figure}[t]
\begin{center}
\includegraphics[scale=.5,angle=0]{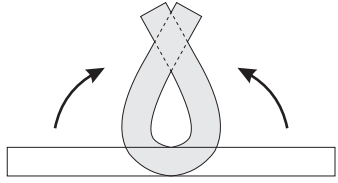}
\end{center}
\caption{\footnotesize A mapping with a compatible displacement gradient, which is not an embedding.}
\label{LocDif}
\end{figure}
%-----------------------------
%-----------------------------

\begin{rem}\label{NotEmb} This theorem does not guarantee that the displacement gradient is induced by a motion of $\mathcal{B}$, i.e. $\varphi$ is not an embedding, in general. For example, consider the mapping depicted in Fig. \ref{LocDif} which is not injective. This mapping is a local diffeomorphism, its tangent map is bijective at all points, and its displacement gradient satisfies the above condition. Also note that in contrary to Theorem \ref{ThmLoc_ConCur}, $\varphi$ is unique only up to rigid body translations and not rigid body rotations. This is a direct consequence of the fact that $\mathrm{H}^{0}_{dR}(\mathcal{B})\approx\mathbb{R}$, for any connected manifold $\mathcal{B}$.
\end{rem}

If $\mathrm{H}^{1}_{dR}(\mathcal{B})$ is finite dimensional, then the integral condition in the above theorem merely needs to be checked for a finite number of closed curves and Theorem \ref{CompF23D} is equivalent to Proposition 2.1 of \citep{Yavari2013}. In contrary to the compatibility equations in terms of $\boldsymbol{C}$, by using the notion of displacement, we are explicitly using the linear structure of $\mathbb{R}^{n}$ for writing the compatibility equations in terms of the displacement gradient.

\begin{rem}
The Green deformation tensor does not induce any linear complex for describing the kinematics of $\varphi$. Let $(\mathcal{S},\boldsymbol{g})$ have constant curvature $\widehat{\mathrm{k}}$ and let $C(\mathcal{B},\mathcal{S})$ and $\Gamma_{M}(S^{2}T^{\ast}\mathcal{B})$ be the spaces of smooth embeddings of $\mathcal{B}$ into $\mathcal{S}$ and Riemannian metrics on $\mathcal{B}$, respectively. Consider the operators $\mathrm{D}_{M}:C(\mathcal{B},\mathcal{S})\rightarrow\Gamma_{M}(S^{2}T^{\ast}\mathcal{B})$, $\mathrm{D}_{M}(\varphi):=\varphi^{\ast}\boldsymbol{g}$, and $\mathrm{D}_{R}:\Gamma_{M}(S^{2}T^{\ast}\mathcal{B})\rightarrow\Gamma(S^{2}(\Lambda^{2}T^{\ast}\mathcal{B}))$ given by
%-----------------------------
\begin{equation}
\begin{aligned}
\big(\mathrm{D}_{R}(\boldsymbol{C})\big)&(\boldsymbol{X}_{\!1},\boldsymbol{X}_{\!2},\boldsymbol{X}_{\!3},\boldsymbol{X}_{\!4})= \boldsymbol{\mathcal{R}}^{\boldsymbol{C}}(\boldsymbol{X}_{\!1},\boldsymbol{X}_{\!2},\boldsymbol{X}_{\!3},\boldsymbol{X}_{\!4}) \\&- \widehat{\mathrm{k}}\boldsymbol{C}(\boldsymbol{X}_{\!3},\boldsymbol{X}_{\!2})\boldsymbol{C}(\boldsymbol{X}_{\!1},\boldsymbol{X}_{\!4}) + \widehat{\mathrm{k}}\boldsymbol{C}(\boldsymbol{X}_{\!3},\boldsymbol{X}_{\!1})\boldsymbol{C}(\boldsymbol{X}_{\!2},\boldsymbol{X}_{\!4}). \nonumber
\end{aligned}
\end{equation}
%-----------------------------
The compatibility equation (\ref{com_con_sc}) implies that $\mathrm{D}_{R}\circ \mathrm{D}_{M}=0$. However, note that the sequence of operators
%-----------------------------
\begin{equation}\label{KinmNonlinS}
\xymatrix{
C(\mathcal{B},\mathcal{S}) \ar[r]^-{\mathrm{D}_{M}} &\Gamma_{M}(S^{2}T^{\ast}\mathcal{B}) \ar[r]^-{\mathrm{D}_{R}} &\Gamma(S^{2}(\Lambda^{2}T^{\ast}\mathcal{B})),  } 
\end{equation}
%-----------------------------
is not a linear complex as the underlying spaces and operators are not linear.\footnote{The complex \eqref{KinmNonlinS} was suggested to us by Marino Arroyo.} The operator $\mathbf{curl}\circ\mathbf{curl}$ of the linear elasticity complex is related to the nonlinear compatibility equations in terms of $\boldsymbol{C}$. Note that the kinematics part of this complex is not the linearization of the kinematics part of the $\boldsymbol{\mathsf{GCD}}$ complex.
\end{rem}

%-----------------------------
%-----------------------------
\subsubsection{Shells}

Let $(\mathcal{S},\boldsymbol{g})$ be an orientable $n$-manifold with constant curvature $\widehat{\mathrm{k}}$. We will derive the compatibility equations for motions of hypersurfaces in $\mathcal{S}$, i.e. motions of $(n-1)$-dimensional submanifolds of $\mathcal{S}$. We first tersely review some preliminaries of submanifold theory, see \citep{Docarmo1992,Ten1971} for more details.

Suppose $(\mathcal{B},\boldsymbol{G})$ is a connected orientable submanifold of $\mathcal{S}$, where $\boldsymbol{G}$ is induced by $\boldsymbol{g}$. Let $\nabla$ and $\nabla^{\boldsymbol{g}}$ be the associated Levi-Civita connections of $\mathcal{B}$ and $\mathcal{S}$, respectively, and let $\boldsymbol{X}_{\!1},\dots,\boldsymbol{X}_{\!4}\in\mathfrak{X}(\mathcal{B})$. We have the decomposition $T_{X}\mathcal{S}=T_{X}\mathcal{B}\oplus (T_{X}\mathcal{B})^{\perp}$, $\forall X\in\mathcal{B}$, where $(T_{X}\mathcal{B})^{\perp}$ is the normal complement of $T_{X}\mathcal{B}$ in $T\mathcal{S}$. Any vector field $\boldsymbol{X}_{\!1}$ on $\mathcal{B}$ can be locally extended to a vector field $\widetilde{\boldsymbol{X}}_{\!1}$ on $\mathcal{S}$ and we have $\nabla_{\!\boldsymbol{X}_{\!1}}\!\boldsymbol{X}_{\!2}=(\nabla^{\boldsymbol{g}}_{\!\widetilde{\boldsymbol{X}}_{\!1}}\! \widetilde{\boldsymbol{X}}_{\!2}  )^{\mathrm{T}}$, where $\mathrm{T}$ denotes the tangent component.

The second fundamental form $\boldsymbol{B}$ of $\mathcal{B}$ is defined as $\boldsymbol{B}(\boldsymbol{X}_{\!1},\boldsymbol{X}_{\!2})=\nabla^{\boldsymbol{g}}_{\!\widetilde{\boldsymbol{X}}_{\!1}}\!\widetilde{\boldsymbol{X}}_{\!2}-\nabla_{\!\boldsymbol{X}_{\!1}}\!\boldsymbol{X}_{\!2}$. Let $\boldsymbol{\EuScript{X}}\in\Gamma(T\mathcal{B}^{\perp})=:\mathfrak{X}(\mathcal{B})^{\perp}$. The shape operator of $\mathcal{B}$ is a linear self-adjoint operator $\mathsf{S}_{\boldsymbol{\EuScript{X}}}:T\mathcal{B}\rightarrow T\mathcal{B}$ defined as $\boldsymbol{G}(\mathsf{S}_{\boldsymbol{\EuScript{X}}}(\boldsymbol{X}_{\!1}),\boldsymbol{X}_{\!2})=\boldsymbol{g}(\boldsymbol{B}(\boldsymbol{X}_{\!1},\boldsymbol{X}_{\!2}),\boldsymbol{\EuScript{X}})$. One can show that $\big(\nabla^{\boldsymbol{g}}_{\!\boldsymbol{X}_{\!1}}\!\boldsymbol{\EuScript{X}}\big)^{\mathrm{T}}=-\mathsf{S}_{\boldsymbol{\EuScript{X}}}(\boldsymbol{X}_{\!1})$. It is also possible to define a linear connection $\nabla^{\perp}$ on $T\mathcal{B}^{\perp}$ by $\nabla^{\perp}_{\boldsymbol{X}_{\!1}}\boldsymbol{\EuScript{X}}=(\nabla^{\boldsymbol{g}}_{\!\boldsymbol{X}_{\!1}}\boldsymbol{\EuScript{X}})^{\mathrm{N}}$, where $\mathrm{N}$ denotes the normal component. The normal curvature $\boldsymbol{R}^{\perp}:\mathfrak{X}(\mathcal{B})\times\mathfrak{X}(\mathcal{B})\times\mathfrak{X}(\mathcal{B})^{\perp}\rightarrow\mathfrak{X}(\mathcal{B})^{\perp}$ is the curvature of $\nabla^{\perp}$. One can show that the following relations hold:
%-----------------------------
\begin{align}
\!\boldsymbol{\mathcal{R}}^{\boldsymbol{g}}(\boldsymbol{X}_{1},\boldsymbol{X}_{2},\boldsymbol{X}_{3},\boldsymbol{X}_{4})&=\boldsymbol{\mathcal{R}}(\boldsymbol{X}_{1},\boldsymbol{X}_{2},\boldsymbol{X}_{3},\boldsymbol{X}_{4}) \nonumber \\ 
& \!\!\!\!\!\!\!\!\!\!\!\!\!\!\!\!\!\!\!\!\!\!\!\!\!\!\!\!\!\!\!\!\!\!\! +\!\boldsymbol{g}(\boldsymbol{B}(\boldsymbol{X}_{1},\boldsymbol{X}_{3}),\boldsymbol{B}(\boldsymbol{X}_{2},\boldsymbol{X}_{4}))-\boldsymbol{g}(\boldsymbol{B}(\boldsymbol{X}_{1},\boldsymbol{X}_{4}),\boldsymbol{B}(\boldsymbol{X}_{2},\boldsymbol{X}_{3})), \label{Ga}\\
\boldsymbol{G}([\mathsf{S}_{\boldsymbol{\EuScript{Y}}},\mathsf{S}_{\boldsymbol{\EuScript{X}}}]\boldsymbol{X}_{\!1},\boldsymbol{X}_{\!2})&= \boldsymbol{g}(\boldsymbol{R}^{\boldsymbol{g}}(\boldsymbol{X}_{\!1},\boldsymbol{X}_{\!2})\boldsymbol{\EuScript{X}},\boldsymbol{\EuScript{Y}}) - \boldsymbol{g}(\boldsymbol{R}^{\perp}(\boldsymbol{X}_{\!1},\boldsymbol{X}_{\!2})\boldsymbol{\EuScript{X}},\boldsymbol{\EuScript{Y}}), \label{Ri} \\ 
\boldsymbol{g}(\!\boldsymbol{R}^{\boldsymbol{g}}\!(\!\boldsymbol{X}_{\!1},\!\boldsymbol{X}_{\!2}\!)\!\boldsymbol{X}_{\!3},\!\boldsymbol{\EuScript{X}}\!) &=  \big(\nabla_{\boldsymbol{X}_{\!1}}\boldsymbol{\EuScript{B}}\big)(\boldsymbol{X}_{\!2},\boldsymbol{X}_{\!3},\boldsymbol{\EuScript{X}}) - \big(\nabla_{\boldsymbol{X}_{\!2}}\boldsymbol{\EuScript{B}}\big)(\boldsymbol{X}_{\!1},\boldsymbol{X}_{\!3},\boldsymbol{\EuScript{X}}),\label{Co}
\end{align}
%-----------------------------  
where $\boldsymbol{\EuScript{X}},\boldsymbol{\EuScript{Y}}\in\mathfrak{X}(\mathcal{B})^{\perp}$, $[\mathsf{S}_{\boldsymbol{\EuScript{Y}}},\mathsf{S}_{\boldsymbol{\EuScript{X}}}]=\mathsf{S}_{\boldsymbol{\EuScript{Y}}}\circ\mathsf{S}_{\boldsymbol{\EuScript{X}}}-\mathsf{S}_{\boldsymbol{\EuScript{X}}}\circ \mathsf{S}_{\boldsymbol{\EuScript{Y}}}$, and $\boldsymbol{\EuScript{B}}(\boldsymbol{X}_{\!1},\boldsymbol{X}_{\!2},\boldsymbol{\EuScript{X}})=\boldsymbol{g}(\boldsymbol{B}(\boldsymbol{X}_{\!1},\boldsymbol{X}_{\!2}),\boldsymbol{\EuScript{X}})$, with
%-----------------------------
\begin{align}
\big(\nabla_{\boldsymbol{X}_1}\boldsymbol{\EuScript{B}}\big)(\boldsymbol{X}_2,\boldsymbol{X}_3,\boldsymbol{\EuScript{X}})
&=\boldsymbol{X}_1\big(\boldsymbol{\EuScript{B}}(\boldsymbol{X}_2,\boldsymbol{X}_3,\boldsymbol{\EuScript{X}})\big) - \boldsymbol{\EuScript{B}}(\nabla_{\boldsymbol{X}_1}\boldsymbol{X}_2,\boldsymbol{X}_3,\boldsymbol{\EuScript{X}}) \nonumber\\
& - \boldsymbol{\EuScript{B}}(\boldsymbol{X}_2,\nabla_{\boldsymbol{X}_1}\boldsymbol{X}_3,\boldsymbol{\EuScript{X}}) - \boldsymbol{\EuScript{B}}(\boldsymbol{X}_2,\boldsymbol{X}_3,\nabla^{\perp}_{\boldsymbol{X}_1}\boldsymbol{\EuScript{X}}). \nonumber
\end{align}
%-----------------------------  
The equations (\ref{Ga}), (\ref{Ri}), and (\ref{Co}) are called the Gauss, Ricci, and Codazzi equations, respectively. These equations generalize the compatibility equations of the local theory of surfaces. Let $\dim\mathcal{S}-\dim\mathcal{B}=1$. By using (\ref{com_con_sct}) and the fact that the second fundamental form of hypersurfaces can be expressed as $\boldsymbol{B}(\boldsymbol{X}_{\!1},\boldsymbol{X}_{\!2})=\boldsymbol{g}(\boldsymbol{B}(\boldsymbol{X}_{\!1},\boldsymbol{X}_{\!2}),\boldsymbol{\EuScript{N}})\boldsymbol{\EuScript{N}}$, where $\boldsymbol{\EuScript{N}}$ is the unit normal vector field of $\mathcal{B}$, the Gauss equation can be written as 
%-----------------------------
\begin{align}
& \boldsymbol{\mathcal{R}}(\boldsymbol{X}_1,\boldsymbol{X}_2,\boldsymbol{X}_3,\boldsymbol{X}_4)+
\boldsymbol{G}(\mathsf{S}_{\boldsymbol{\EuScript{N}}}\boldsymbol{X}_3,\boldsymbol{X}_1)\boldsymbol{G}(\mathsf{S}_{\boldsymbol{\EuScript{N}}}\boldsymbol{X}_4,\boldsymbol{X}_2) \nonumber\\
& - \boldsymbol{G}(\mathsf{S}_{\boldsymbol{\EuScript{N}}}\boldsymbol{X}_4,\boldsymbol{X}_1)\boldsymbol{G}(\mathsf{S}_{\boldsymbol{\EuScript{N}}}\boldsymbol{X}_3,\boldsymbol{X}_2)+ \widehat{\mathrm{k}}\boldsymbol{G}(\boldsymbol{X}_1,\boldsymbol{X}_3)\boldsymbol{G}(\boldsymbol{X}_2,\boldsymbol{X}_4) \nonumber  \\
&  - \widehat{\mathrm{k}}\boldsymbol{G}(\boldsymbol{X}_1,\boldsymbol{X}_4)\boldsymbol{G}(\boldsymbol{X}_2,\!\boldsymbol{X}_3) = 0. \nonumber 
\end{align}
%-----------------------------  
For hypersurfaces in an ambient space with constant curvature the Ricci equation becomes vacuous and the Codazzi equation simplifies to read     
%-----------------------------
\begin{equation}\label{Co2}
\nabla_{\boldsymbol{X}_{\!1}}\!\!\left(\mathsf{S}_{\boldsymbol{\EuScript{N}}}(\boldsymbol{X}_{\!2})\right) - \nabla_{\boldsymbol{X}_{\!2}}\!\!\left(\mathsf{S}_{\boldsymbol{\EuScript{N}}}(\boldsymbol{X}_{\!1})\right) = \mathsf{S}_{\boldsymbol{\EuScript{N}}}([\boldsymbol{X}_{\!1},\boldsymbol{X}_{\!2}]). \nonumber
\end{equation}
%-----------------------------   
Let $\varphi:\mathcal{B}\rightarrow\mathcal{S}$ be an orientation-preserving isometric embedding and let $\bar{\boldsymbol{X}}_{\!1}=\varphi_{\ast}\boldsymbol{X}_{\!1}\in\mathfrak{X}(\varphi(\mathcal{B}))$. The extrinsic deformation tensor $\boldsymbol{\theta}\in\Gamma(S^{2}T^{\ast}\mathcal{B})$ is defined as $\boldsymbol{\theta}(\boldsymbol{X}_{\!1},\boldsymbol{X}_{\!2}):=\boldsymbol{g}(\bar{\boldsymbol{B}}(\bar{\boldsymbol{X}}_{\!1},\bar{\boldsymbol{X}}_{\!2}),\bar{\boldsymbol{\EuScript{N}}} )$, where $\bar{\boldsymbol{B}}$ is the second fundamental form of the hypersurface $\varphi(\mathcal{B})\subset\mathcal{S}$ with the unit normal vector field $\bar{\boldsymbol{\EuScript{N}}}$ and the induced metric $\bar{\boldsymbol{g}}:=\boldsymbol{g}|_{\varphi(\mathcal{B})}$. Let $\boldsymbol{C}=\varphi^{\ast}\bar{\boldsymbol{g}}$ be the Green deformation tensor. The pull-back of the Gauss equation on $\left(\varphi(\mathcal{B}), \bar{\boldsymbol{g}} \right)$ by $\varphi$ reads
%-----------------------------
\begin{equation}
\begin{aligned}\label{Ga3}
&\boldsymbol{\mathcal{R}}^{\boldsymbol{C}}(\boldsymbol{X}_{\!1},\boldsymbol{X}_{\!2},\boldsymbol{X}_{\!3},\boldsymbol{X}_{\!4})+
\boldsymbol{\theta}(\boldsymbol{X}_{\!1},\boldsymbol{X}_{\!3})\boldsymbol{\theta}(\boldsymbol{X}_{\!2},\boldsymbol{X}_{\!4}) - \boldsymbol{\theta}(\boldsymbol{X}_{\!1},\boldsymbol{X}_{\!4})\boldsymbol{\theta}(\boldsymbol{X}_{\!2},\boldsymbol{X}_{\!3})  \\&+
\widehat{\mathrm{k}}\boldsymbol{C}(\boldsymbol{X}_{\!1},\boldsymbol{X}_{\!3})\boldsymbol{C}(\boldsymbol{X}_{\!2},\boldsymbol{X}_{\!4}) - \widehat{\mathrm{k}}\boldsymbol{C}(\boldsymbol{X}_{\!1},\boldsymbol{X}_{\!4})\boldsymbol{C}(\boldsymbol{X}_{\!2},\boldsymbol{X}_{\!3}) =0. 
\end{aligned}
\end{equation}
%----------------------------- 
The pull-back of the Codazzi equation on $\left(\varphi(\mathcal{B}), \bar{\boldsymbol{g}} \right)$ by $\varphi$ reads
%-----------------------------
\begin{equation}\label{Co3}
\left(\nabla^{\boldsymbol{C}}_{\boldsymbol{X}_{\!1}}\boldsymbol{\theta}\right)(\boldsymbol{X}_{\!2},\boldsymbol{X}_{\!3}) = \left(\nabla^{\boldsymbol{C}}_{\boldsymbol{X}_{\!2}}\boldsymbol{\theta}\right)(\boldsymbol{X}_{\!1},\boldsymbol{X}_{\!3}),
\end{equation}
%-----------------------------  
i.e. the $({}^{0}_{3})$-tensor $\nabla^{\boldsymbol{C}}\boldsymbol{\theta}$ defined by $\left(\nabla^{\boldsymbol{C}}\boldsymbol{\theta}\right)(\boldsymbol{X}_{\!1},\boldsymbol{X}_{\!2},\boldsymbol{X}_{\!3}):= \left(\nabla^{\boldsymbol{C}}_{\boldsymbol{X}_{\!1}}\boldsymbol{\theta}\right)(\boldsymbol{X}_{\!2},\boldsymbol{X}_{\!3})$, is completely symmetric.

The compatibility problem for motions of hypersurfaces in terms of $\boldsymbol{C}$ and $\boldsymbol{\theta}$ can be stated as follows: Given a metric $\boldsymbol{C}\in\Gamma(S^{2}T^{\ast}\mathcal{B})$ on $\mathcal{B}$ and a symmetric tensor $\boldsymbol{\theta}\in\Gamma(S^{2}T^{\ast}\mathcal{B})$, determine the necessary and sufficient conditions for the existence of an isometric embedding $\varphi:\mathcal{B}\rightarrow\mathcal{S}$ such that $\boldsymbol{C}=\varphi^{\ast}\bar{\boldsymbol{g}}$, and $\boldsymbol{\theta}(\boldsymbol{X}_{\!1},\boldsymbol{X}_{\!2})=\boldsymbol{g}(\bar{\boldsymbol{B}}(\bar{\boldsymbol{X}}_{\!1},\bar{\boldsymbol{X}}_{\!2}),\bar{\boldsymbol{\EuScript{N}}} )$. The reason for including $\boldsymbol{\theta}$ in the compatibility problem is that we want surfaces with identical deformation tenors to be unique up to isometries of the ambient space. This criterion cannot be satisfied if we only consider $\boldsymbol{C}$. For example, consider isometric deformations of a plane in $\mathbb{R}^{3}$ into portions of cylinders with different radii as shown in Fig. \ref{IsoSamp}. All these motions induce the same $\boldsymbol{C}$, but obviously cylinders with different radii cannot be mapped into each other using rigid body motions in $\mathbb{R}^{3}$. The above discussion together with some standard results of submanifold theory (e.g. see \citet[Chapter 2]{IvLa2003} or \citet{KN1969}) give us the following theorem. 
%-----------------------------
%-----------------------------
\begin{figure}[t]
\begin{center}
\includegraphics[scale=.5,angle=0]{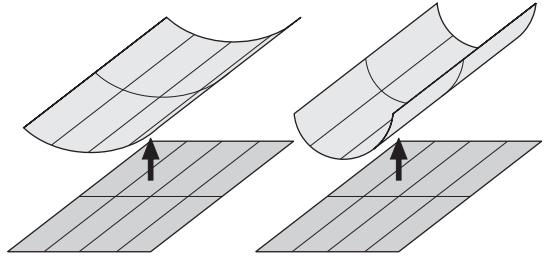}
\end{center}
\caption{\footnotesize Two isometric embeddings of a plane into $\mathbb{R}^{3}$. The resulting surfaces are cylinders with different radii and both motions have the same Green deformation tensor $\boldsymbol{C}$.}
\label{IsoSamp}
\end{figure}
%-----------------------------
%-----------------------------

\begin{thm}\label{ShellComp}Suppose $\mathcal{S}$ is a complete, simply-connected $n$-manifold with constant curvature and $\mathcal{B}$ is a connected hypersurface in $\mathcal{S}$. The deformation tensors $\boldsymbol{C}$ and $\boldsymbol{\theta}$ induced by an embedding $\varphi:\mathcal{B}\rightarrow\mathcal{S}$ satisfy (\ref{Ga3}) and (\ref{Co3}). Conversely, if a Riemannian metric $\boldsymbol{C}$ on $\mathcal{B}$ and a symmetric tensor $\boldsymbol{\theta}\in\Gamma(S^{2}T^{\ast}\mathcal{B})$ satisfy (\ref{Ga3}) and (\ref{Co3}), for each $X\in\mathcal{B}$, there is an open neighborhood $\mathcal{U}_{X}\subset\mathcal{B}$ of $X$ and a local embedding $\varphi_{X}:\mathcal{U}_{X}\rightarrow\mathcal{S}$, such that $\boldsymbol{C}|_{\mathcal{U}_{X}}$ and $\boldsymbol{\theta}|_{\mathcal{U}_{X}}$ are the deformation tensors of $\varphi_{X}$. The embedding $\varphi_{X}$ is unique up to isometries of $\mathcal{S}$.    
\end{thm}

If $\mathcal{B}$ is also simply-connected, then (\ref{Ga3}) and (\ref{Co3}) imply that there exists a global embedding $\varphi:\mathcal{B}\rightarrow\mathcal{S}$, which is unique up to isometries of $\mathcal{S}$ \citep{Ten1971}. The relations (\ref{Ga3}) and (\ref{Co3}) generalize the classical compatibility equations of $2$D surfaces in $\mathbb{R}^{3}$ discussed in \citet{CiarletMardare2008}.

One can exploit the \emph{\textsf{GC}} complex for writing the compatibility equations for motions of shells in $\mathbb{R}^{3}$ in terms of the displacement gradient. Using the same notation used in Theorem \ref{ExactGradSSurf2pTen2D}, the upshot can be stated as follows.

\begin{thm}\label{CompF2DSurfaces} Suppose $\mathcal{B}\subset\mathbb{R}^{3}$ is a connected $2$D surface. Given $\boldsymbol{\kappa}=(\boldsymbol{\kappa}^{1},\boldsymbol{\kappa}^{2},\boldsymbol{\kappa}^{3})\in\Omega^{1}(\mathcal{B};\mathbb{R}^{3})$, there is a smooth mapping $\varphi:\mathcal{B}\rightarrow\mathbb{R}^{3}$ with displacement gradient $\mathbf{J}^{-1}_{1}(\boldsymbol{\kappa})$ if and only if 
%-----------------------------
\begin{equation}\label{CompF2DSurfacesEq}
\boldsymbol{d}\boldsymbol{\kappa}=0, \text{ and }~~~\int_{\ell}\boldsymbol{\kappa}(\boldsymbol{t}_{\ell})dS=0, ~\forall\ell\subset\mathcal{B}. \nonumber 
\end{equation}
%----------------------------- 
The mapping $\varphi$ is unique up to rigid body translations in $\mathbb{R}^{3}$.   
\end{thm}

It is straightforward to extend the above theorem to hypersurfaces in $\mathbb{R}^{n}$. Also note that a discussion similar to Remark \ref{NotEmb} shows that the mapping $\varphi$ in Theorem \ref{CompF2DSurfaces} is not necessarily an embedding and unlike Theorem \ref{ShellComp}, $\varphi$ is unique only up to rigid body translations in $\mathbb{R}^{3}$.

%----------------------------- 
%----------------------------- 
\subsubsection{Linearized elasticity on curved manifolds}

The operator $\mathrm{D}_{1}:\Gamma(S^{2}T^{\ast}\mathcal{B})\rightarrow \Gamma(S^{2}(\Lambda^{2}T^{\ast}\mathcal{B}))$ in the Calabi complex expresses the compatibility equations for the linear strain on the $n$-manifold $(\mathcal{B},\boldsymbol{G})$ with constant curvature $\mathrm{k}$. Note that $\mathrm{D}_{1}$ is obtained by linearizing the Riemannian curvature, and therefore, it is related to the compatibility equations for $\boldsymbol{C}$. Next, we write $\mathrm{D}_{1}$ in a local coordinate system. To simplify the calculations, we use the normal coordinate system $\{X^{I}\}$ in the following sense: At any point $X$ of $\mathcal{B}$, there is a local coordinate system $\{X^{I}\}$ centered at $X$ such that $\nabla_{\boldsymbol{E}_{I}}\boldsymbol{E}_{J}=0$, at $X$, where $\nabla$ is the Levi-Civita connection and $\{\boldsymbol{E}_{I}\}$ is the local basis induced by $\{X^{I}\}$ for $T\mathcal{B}$, which is orthonormal at $X$ \citep{KN1963}. The Cartesian coordinate system of $\mathbb{R}^{n}$ is a global normal coordinate system for the Euclidean space. Let $\boldsymbol{e}\in\Gamma(S^{2}T^{\ast}\mathcal{B})$. In a normal coordinate system $\{X^{I}\}$ at $X$, it is straightforward to verify that
%-----------------------------
\begin{equation}\label{LComp_temp1}
\begin{aligned}
& \left(\nabla_{\boldsymbol{E}_{I}} \nabla_{\boldsymbol{E}_{K}}\boldsymbol{e}\right)(\boldsymbol{E}_{J},\boldsymbol{E}_{L}) =  \boldsymbol{E}_{I} \left(\boldsymbol{E}_{K} \left(\boldsymbol{e}(\boldsymbol{E}_{J},\boldsymbol{E}_{L}) \right)\right)  -  \boldsymbol{E}_{I}\big(\boldsymbol{e}(\nabla_{\boldsymbol{E}_{K}}\boldsymbol{E}_{J},\boldsymbol{E}_{L}) \\ 
&+ \boldsymbol{e}(\boldsymbol{E}_{J},\nabla_{\boldsymbol{E}_{K}}\boldsymbol{E}_{L})\big). \nonumber
\end{aligned}
\end{equation}
%-----------------------------
We have $\nabla_{\boldsymbol{E}_{I}}\boldsymbol{E}_{J}=\Gamma^{K}_{IJ}\boldsymbol{E}_{K}$, where $\Gamma^{K}_{IJ}$'s are the Christoffel symbols of $\nabla$. The linear compatibility equations can be written as $\mathrm{D}_{1}(\boldsymbol{e})=0$. By using the above relation, the compatibility equation at $X$ corresponding to the component $(\mathrm{D}_{1}(\boldsymbol{e}))_{IJKL}$ reads
%-----------------------------
\begin{equation}\label{LComp_temp2}
\begin{aligned}
&\frac{\partial^{2}e_{JL}}{\partial X^{I}\partial X^{K}} \!+\! \frac{\partial^{2}e_{IK}}{\partial X^{J}\partial X^{L}}\! -\! \frac{\partial^{2}e_{JK}}{\partial X^{I}\partial X^{L}} \!-\! \frac{\partial^{2}e_{IL}}{\partial X^{J}\partial X^{K}}\! +\! \left( \frac{\partial \Gamma^{M}_{LJ}}{\partial X^{I}}\! -\! \frac{\partial \Gamma^{M}_{LI}}{\partial X^{J}} \right) e_{MK} \\ &+\! \left( \frac{\partial \Gamma^{M}_{KI}}{\partial X^{J}} \!-\! \frac{\partial \Gamma^{M}_{KJ}}{\partial X^{I}} \right) e_{ML} \!+\! \mathrm{k}\big\{\text{\textdelta}_{JK}e_{IL}\!-\!\text{\textdelta}_{IK}e_{JL} \!-\! \text{\textdelta}_{JL}e_{IK}\!+\! \text{\textdelta}_{IL}e_{JK}\big\}\!=\!0. \nonumber
\end{aligned}
\end{equation}
%-----------------------------
If $\mathcal{B}\subset\mathbb{R}^{n}$ and $\{X^{I}\}$ is the Cartesian coordinate system, one recovers the classical expression $\mathbf{curl}\circ\mathbf{curl}\;\boldsymbol{e}=0$. For $n=2$, there is only one compatibility equation corresponding to $(\mathrm{D}_{1}(\boldsymbol{e}))_{1212}$:
%-----------------------------
\begin{align}
& \frac{\partial^{2}e_{11}}{\partial\! X^{2}\partial X^{2}} - 2 \frac{\partial^{2}e_{12}}{\partial  X^{1}\partial X^{2}} + \frac{\partial^{2}e_{22}}{\partial  X^{1}\partial X^{1}} + \left(\frac{\partial \Gamma^{M}_{11}}{\partial X^{2}} - \frac{\partial \Gamma^{M}_{12}}{\partial X^{1}} \right) e_{M2} \nonumber\\
& + \left(\frac{\partial \Gamma^{M}_{22}}{\partial X^{1}} - \frac{\partial \Gamma^{M}_{21}}{\partial X^{2}} \right) e_{M1} -\mathrm{k}(e_{11}+e_{22})=0.\nonumber
\end{align}
%-----------------------------
For $n=3$, we have $6$ compatibility equations corresponding to $(\mathrm{D}_{1}(\boldsymbol{e}))_{1212}$, $(\mathrm{D}_{1}(\boldsymbol{e}))_{1223}$, $(\mathrm{D}_{1}(\boldsymbol{e}))_{1313}$, $(\mathrm{D}_{1}(\boldsymbol{e}))_{2113}$, $(\mathrm{D}_{1}(\boldsymbol{e}))_{2323}$, and $(\mathrm{D}_{1}(\boldsymbol{e}))_{3123}$.

As an example, let us write the compatibility equation on the $2$-sphere with radius $\EuScript{R}$ and $\mathrm{k}=1/\EuScript{R}^{2}$. We choose the spherical coordinate system with $(X^{1},X^{2}):=(\theta,\phi)$, with $G_{11}=\EuScript{R}^{2}\sin^{2} \phi$, $G_{12}=0$, and $G_{22}=\EuScript{R}^{2}$. The nonzero Christoffel symbols are $\Gamma^{2}_{11}=-\frac{1}{2}\sin 2\phi$, and $\Gamma^{1}_{12}=\Gamma^{1}_{21}=\cot\phi$. Note that $(\theta,\phi)$ is an orthogonal coordinate system but it is not a normal coordinate system at any point. Therefore, we must use the general form of the compatibility equations given in (\ref{Lin_Comp_thm}). Using the relations $\nabla_{\boldsymbol{E}_{1}}\boldsymbol{E}_{1}=\Gamma^{2}_{11}\boldsymbol{E}_{2}$, $\nabla_{\boldsymbol{E}_{1}}\boldsymbol{E}_{2}=\nabla_{\boldsymbol{E}_{2}}\boldsymbol{E}_{1} = \Gamma^{1}_{12}\boldsymbol{E}_{1}$, and $\nabla_{\boldsymbol{E}_{2}}\boldsymbol{E}_{2}=0$, one obtains the following compatibility equation:
%-----------------------------
\begin{equation}\label{ExTT1}
\begin{aligned}
&\frac{\partial^{2}e_{11}}{\partial X^{\!2} \partial X^{\!2}} -2 \frac{\partial^{2}e_{12}}{\partial X^{1}\partial X^{2}} + \frac{\partial^{2}e_{22}}{\partial X^{1}\partial X^{1}}-\left(\cot X^{2}\right) \frac{\partial e_{11}}{\partial X^{2}}  \\
& -\frac{1}{2}\left(\sin 2X^{2}\right) \frac{\partial e_{22}}{\partial X^{2}} + 2\left(\cot^{2}X^{2}\right) e_{11}=0.\nonumber
\end{aligned}
\end{equation}
%-----------------------------
The components $e_{IJ}$ are not the conventional components $e_{\theta\theta}$, $e_{\theta\phi}$, and $e_{\phi\phi}$ of the linear strain in the spherical coordinate system as $\boldsymbol{E}_{1}$ and $\boldsymbol{E}_{2}$ are not unit vector fields. Since $e_{11}= \EuScript{R}^{2}\sin^{2}\! \phi\; e_{\theta\theta}$, $e_{12}= \EuScript{R}^{2}\sin \phi\; e_{\theta\phi}$, and $e_{22}= \EuScript{R}^{2} e_{\phi\phi}$, the above equation can be written as 
%-----------------------------
\begin{equation}
\begin{aligned}
	& \sin^{2}\phi\frac{\partial^{2}e_{\theta\theta}}{\partial \phi^{2} } -2\frac{\partial^{2}\left(e_{\theta\phi}\sin\phi \right)}{\partial \theta\partial \phi} + \frac{\partial^{2}e_{\phi\phi}}{\partial\theta^{2}} +\frac{3}{2}\sin2\phi \frac{\partial e_{\theta\theta}}{\partial \phi} \\
	& -\frac{1}{2}\sin 2\phi\frac{\partial e_{\phi\phi}}{\partial \phi} + (\sin2\phi-1)e_{\theta\theta}=0. \nonumber
\end{aligned}
\end{equation}
%-----------------------------

Note that instead of using the Calabi complex and the linearization of the Riemannian curvature, it is possible to derive the compatibility equations of the linear strain on manifolds with constant curvature by a less-systematic elimination approach discussed in \citep[\S11]{TRUESDELL1959}.

%-----------------------------
%-----------------------------
\subsection{Stress Functions}

Next, we study the applications of the complexes we derived earlier to the existence of stress functions. Let $\varphi:\mathcal{B}\rightarrow\mathcal{S}=\mathbb{R}^{3}$ be a motion of a $3$-manifold $\mathcal{B}\subset{\mathbb{R}^{3}}$ and suppose $\boldsymbol{\sigma}\in\Gamma(S^{2}T\varphi(\mathcal{B}))$ is the associated symmetric Cauchy stress tensor. Since $(\varphi(\mathcal{B}),\boldsymbol{g})$ is a flat manifold, one obtains a diagram similar to (\ref{3DElasCal}) on $(\varphi(\mathcal{B}),\boldsymbol{g})$. Potentials induced by the operator $\mathbf{curl}\circ\mathbf{curl}$ for $\boldsymbol{\sigma}$ are the Beltrami stress functions. The necessary and sufficient conditions for the existence of Beltrami stress functions are given by (\ref{ExactcurlcurlP}): $\boldsymbol{\sigma}$ must be equilibrated and the resultant forces and moments on any closed surface in $\varphi(\mathcal{B})$ must vanish. Such a stress tensor is called totally self-equilibrated \citep{Gurtin1963}. Note that the operator $\mathbf{grad}^{s}$ in the linear elasticity complex on $(\varphi(\mathcal{B}),\boldsymbol{g})$ does not have any obvious physical interpretation.

If $\boldsymbol{\sigma}$ is not symmetric, one can obtain $\mathbf{curl}^{\mathsf{T}}$-stress functions for $\boldsymbol{\sigma}$ by considering the $\boldsymbol{\mathsf{gcd}}$ complex on $(\varphi(\mathcal{B}),\boldsymbol{g})$, i.e. $\mathbf{curl}^{\mathsf{T}}$-stress functions are potentials induced by $\mathbf{curl}^{\mathsf{T}}$. Theorem \ref{ExactcurlTTen} implies that $\boldsymbol{\sigma}$ admits a $\mathbf{curl}^{\mathsf{T}}$-stress function if and only if $\mathbf{div}\,\boldsymbol{\sigma}=0$, and the resultant force on any closed surface in $\varphi(\mathcal{B})$ vanishes. If the closure $\bar{\mathcal{B}}$ of $\mathcal{B}$ is a compact subset of $\mathbb{R}^{3}$, then Theorem \ref{ExactcurlTTen} would be identical to Theorem 2.2 of \citep{Carlson1967}. If $\boldsymbol{\sigma}$ admits a Beltrami stress function, then it also admits a $\mathbf{curl}^{\mathsf{T}}$-stress function. Unlike Beltrami stress functions that are symmetric tensors, even if $\boldsymbol{\sigma}$ is symmetric, $\mathbf{curl}^{\mathsf{T}}$-stress functions are not necessarily symmetric.

For the $2$D case $\mathcal{B}\subset\mathbb{R}^{2}$, Airy stress functions and \textbf{s}-stress functions for symmetric and non-symmetric Cauchy stress tensors are induced by the complex (\ref{2DDivElasCal}) and the $\boldsymbol{\mathsf{sd}}$ complex, respectively. Note that if $\bar{\mathcal{B}}$ is compact, then the complex (\ref{2DDivElasCal}) and the $\boldsymbol{\mathsf{sd}}$ complex are the dual complexes of the $2D$ linear elasticity complex and the $\boldsymbol{\mathsf{gc}}$ complex with respect to the proper $L^{2}$-inner products \citep{AngoshtariYavari2014II}.

In the above discussions, by replacing $(\varphi(\mathcal{B}),\boldsymbol{g})$ with the flat manifold $(\mathcal{B},\boldsymbol{C})$ together with its global orthonormal coordinate system endowed with the Cartesian coordinate system of $\varphi(\mathcal{B})$, one obtains stress functions for the second Piola-Kirchhoff stress tensor $\boldsymbol{S}$ as well. In summary, we observe that the complexes for $\boldsymbol{\sigma}$ and $\boldsymbol{S}$ only describe the kinetics of motion. 

The $\boldsymbol{\mathsf{GCD}}$ complex allows one to introduce $\mathbf{Curl}^{\mathsf{T}}$-stress functions for the first Piola-Kirchhoff stress tensor $\boldsymbol{P}\in\Gamma(T\varphi(\mathcal{B})\otimes T\mathcal{B})$. More specifically, Theorem \ref{ExactGradCurlTPTen} implies that $\boldsymbol{P}$ admits a $\mathbf{Curl}^{\mathsf{T}}$-stress function if and only if $\mathbf{Div}\,\boldsymbol{P}=0$, and the resultant force induced by $\boldsymbol{P}$ on any closed surfaces is zero. This together with Theorem \ref{CompF23D} show that the $\boldsymbol{\mathsf{GCD}}$ complex describes both the kinematics and the kinetics of motion. Similarly, one can define $\mathbf{S}$-stress functions for $2$-manifolds  by using the $\boldsymbol{\mathsf{SD}}$ complex.     

As mentioned in Remark \ref{IndNumCo}, the dimensions of the cohomology groups of the de Rham complex determine the number of independent closed curves and surfaces that one requires in the integral conditions for the existence of potentials. Suppose $\mathrm{H}^{k}_{dR}(\mathcal{B})$ is finite dimensional. Then, for the linear and nonlinear compatibility problems for both $2$- and $3$-manifolds, one merely needs to use $b_{1}(\mathcal{B})$ independent closed curves. The same is true for the existence of stress functions on $2$-manifolds. For $3$-manifolds, we need to consider $b_{2}(\mathcal{B})$ independent closed surfaces. Here, independent closed curves and surfaces are those that induce distinct cohomology classes in $\mathrm{H}^{k}_{dR}(\mathcal{B})$.

%-----------------------------
%-----------------------------
\subsection{Further Applications}

In this paper, we assumed that a body $\mathcal{B}$ is an arbitrary submanifold of $\mathbb{R}^{n}$, e.g. it can be unbounded or has infinite-dimensional de Rham cohomologies. We also assumed that all sections and mappings on $\mathcal{B}$ are $C^{\infty}$. One way to relax this smoothness assumption is to impose certain restrictions on the topology of $\mathcal{B}$. In particular, suppose $\mathcal{B}$ is the interior of a compact manifold $\bar{\mathcal{B}}$. Then, $\bar{\mathcal{B}}$ is a compact manifold with boundary and hence all $\mathrm{H}^{k}_{dR}(\bar{\mathcal{B}})$'s are finite dimensional. The compactness of $\bar{\mathcal{B}}$ allows one to define $L^{2}$-inner products for smooth tensors on $\bar{\mathcal{B}}$ and the completion of smooth tensors with respect to these inner products gives us some Sobolev spaces that contain less smooth tensors as well. Using these Sobolev spaces, one can extend smooth complexes discussed here to more general Hilbert complexes.

As was discussed in \citep{AngoshtariYavari2014II}, one can use the corresponding Hilbert complexes to introduce Hodge-type and Helmholtz-type orthogonal decompositions for second-order tensors. Moreover, one can also include Dirichlet boundary conditions in the compatibility problem. On the other hand, these Hilbert complexes can provide suitable solution spaces for mixed formulations for nonlinear elasticity and inelasticity in terms of the displacement gradient and the first Piola-Kirchhoff stress tensor. Similar to the numerical schemes developed in \citep{Arnold2006, Arnold2010} for the Laplace and the linear elasticity equations, such mixed formulations may provide numerical schemes that are compatible with the topology of the underlying bodies.

\paragraph{Acknowledgments.} We benefited from discussions with Marino Arroyo. AA benefited from discussions with Andreas \v{C}ap and Mohammad Ghomi. We are grateful to anonymous reviewers whose useful comments significantly improved the presentation of this paper. This research was partially supported by AFOSR -- Grant No. FA9550-12-1-0290 and NSF -- Grant No. CMMI 1042559 and CMMI 1130856.

\bibliographystyle{plainnat}
\bibliography{biblio}

\begin{thebibliography}{31}
\providecommand{\natexlab}[1]{#1}
\providecommand{\url}[1]{\texttt{#1}}
\expandafter\ifx\csname urlstyle\endcsname\relax
  \providecommand{\doi}[1]{doi: #1}\else
  \providecommand{\doi}{doi: \begingroup \urlstyle{rm}\Url}\fi

\bibitem[Ambrose(1956)]{Ambrose1956}
W.~Ambrose.
\newblock Parallel translation of {R}iemannian curvature.
\newblock \emph{Ann. of Math.}, 64:\penalty0 337--363, 1956.

\bibitem[Angoshtari and Yavari()]{AngoshtariYavari2014II}
A.~Angoshtari and A.~Yavari.
\newblock Hilbert complexes, orthogonal decompositions, and potentials for
  nonlinear continua.
\newblock \emph{Submitted}.

\bibitem[Arnold et~al.(2006)Arnold, Falk, and Winther]{Arnold2006}
D.~N. Arnold, R.~S. Falk, and R.~Winther.
\newblock Finite element exterior calculus, homological techniques, and
  applications.
\newblock \emph{Acta Numerica}, 15:\penalty0 1--155, 2006.

\bibitem[Arnold et~al.(2010)Arnold, Falk, and Winther]{Arnold2010}
D.~N. Arnold, R.~S. Falk, and R.~Winther.
\newblock Finite element exterior calculus: from {H}odge theory to numerical
  stability.
\newblock \emph{Bul. Am. Math. Soc.}, 47:\penalty0 281--354, 2010.

\bibitem[Baston and Eastwood(1989)]{BastonEastwood1989}
R.~J. Baston and M.~G. Eastwood.
\newblock \emph{The Penrose Transformation: its Interaction with Representation
  Theory}.
\newblock Oxford University Press, 1989.

\bibitem[Bott and Tu(2010)]{BottTu2010}
R.~Bott and L.~W. Tu.
\newblock \emph{Differential Forms in Algebraic Topology}.
\newblock Springer-Verlag, New York, 2010.

\bibitem[Bredon(1993)]{Bredon1993}
G.~E. Bredon.
\newblock \emph{Topology and Geometry}.
\newblock Springer-Verlog, New York, 1993.

\bibitem[Calabi(1961)]{Calabi1961}
E.~Calabi.
\newblock On compact {R}iemannian manifolds with constant curvature {I}.
\newblock In \emph{Differential Geometry}, pages 155--180. Proc. Symp. Pure
  Math. vol. III, Amer. Math. Soc., 1961.

\bibitem[Carlson(1967)]{Carlson1967}
D.~E. Carlson.
\newblock On {G}{\"{u}}nther's stress functions for couple stresses.
\newblock \emph{Q. Appl. Math.}, 25:\penalty0 139--146, 1967.

\bibitem[Carmo(1992)]{Docarmo1992}
{M. do} Carmo.
\newblock \emph{Riemannian Geometry}.
\newblock Birkh{\"{a}}user, Boston, 1992.

\bibitem[Ciarlet et~al.(2008)Ciarlet, Gratie, and Mardare]{CiarletMardare2008}
P.~G. Ciarlet, L.~Gratie, and C.~Mardare.
\newblock A new approach to the fundamental theorem of surface theory.
\newblock \emph{Arch. Rat. Mech. Anal.}, 188:\penalty0 457--473, 2008.

\bibitem[Eastwood(2000)]{Eastwood2000}
M.~G. Eastwood.
\newblock A complex from linear elasticity.
\newblock pages 23--29. Rend. Circ. Mat. Palermo, Serie II, Suppl. 63, 2000.

\bibitem[Gasqui and Goldschmidt(1983)]{GaGo1983}
J.~Gasqui and H.~Goldschmidt.
\newblock D{\'{e}}formations infinit{\'{e}}simales des espaces {R}iemanniens
  localement sym{\'{e}}triques. {I}.
\newblock \emph{Adv. in Math.}, 48:\penalty0 205--285, 1983.

\bibitem[Gasqui and Goldschmidt(2004)]{GaGo2004}
J.~Gasqui and H.~Goldschmidt.
\newblock \emph{Radon Transforms and the Rigidity of the {G}rassmannians}.
\newblock Princeton University Press, Princeton, 2004.

\bibitem[Geymonat and Krasucki(2009)]{GeymonatKrasucki2009}
G.~Geymonat and F.~Krasucki.
\newblock {H}odge decomposition for symmetric matrix fields and the elasticity
  complex in {L}ipschitz domains.
\newblock \emph{Commun. Pure Appl. Anal.}, 8:\penalty0 295--309, 2009.

\bibitem[Gilkey(1984)]{Gilkey1984}
P.~B. Gilkey.
\newblock \emph{Invariance Theory, the Heat Equation, and the {A}tiyah-{S}inger
  Index Theorem}.
\newblock Publish or Perish, Wilmington, 1984.

\bibitem[Gurtin(1963)]{Gurtin1963}
M.~E. Gurtin.
\newblock A generalization of the {B}eltrami stress functions in continuum
  mechanics.
\newblock \emph{Arch. Rational Mech. Anal.}, 13:\penalty0 321--329, 1963.

\bibitem[Hackl and Zastrow(1988)]{HacklZastrow1988}
K.~Hackl and U.~Zastrow.
\newblock On the existence, uniqueness and completeness of displacements and
  stress functions in linear elasticity.
\newblock \emph{J. Elasticity}, 19:\penalty0 3--23, 1988.

\bibitem[Ivey and Landsberg(2003)]{IvLa2003}
T.~A. Ivey and J.~M. Landsberg.
\newblock \emph{{C}artan for Beginners: Differential Geometry via Moving Frames
  and Exterior Differential Systems}.
\newblock American Mathematical Society, Providence, RI, 2003.

\bibitem[Kobayashi and Nomizu(1963)]{KN1963}
S.~Kobayashi and K.~Nomizu.
\newblock \emph{Foundations of Differential Geometry}, volume~1.
\newblock Interscience Publishers, New York, 1963.

\bibitem[Kobayashi and Nomizu(1969)]{KN1969}
S.~Kobayashi and K.~Nomizu.
\newblock \emph{Foundations of Differential Geometry}, volume~2.
\newblock Interscience Publishers, New York, 1969.

\bibitem[Kr{\"{o}}ner(1959)]{Kroner1959}
E.~Kr{\"{o}}ner.
\newblock Allgemeine kontinuumstheorie der versetzungen und eigenspannungen.
\newblock \emph{Arch. Rational Mech. Anal.}, 4:\penalty0 273--334, 1959.

\bibitem[Lee(2012)]{Lee2012}
J.~M. Lee.
\newblock \emph{Introduction to smooth manifolds}.
\newblock Springer, New York, 2012.

\bibitem[Marsden and Hughes(1994)]{MarsdenHughes1994}
J.~E. Marsden and T.~Hughes.
\newblock \emph{Mathematical Foundations of Elasticity}.
\newblock Dover Publications, New York, 1994.

\bibitem[Penrose and Rindler(1984)]{PenroseRindler1984}
R.~Penrose and W.~Rindler.
\newblock \emph{Spinors and Space-time}, volume I: Two-spinor calculus and
  relativistic fields.
\newblock Cambridge University Press, 1984.

\bibitem[Schwarz(1995)]{Schwarz1995}
G.~Schwarz.
\newblock \emph{Hodge Decomposition - A Method for Solving Boundary Value
  Problems (Lecture Notes in Mathematics-1607)}.
\newblock Springer-Verlog, Berlin, 1995.

\bibitem[Srivastava(2008)]{Srivastava2008}
S.~K. Srivastava.
\newblock \emph{General Relativity And Cosmology}.
\newblock Prentice-Hall Of India Pvt. Limited, New Delhi, 2008.

\bibitem[Tenenblat(1971)]{Ten1971}
K.~Tenenblat.
\newblock On isometric immersions of {R}iemannian manifolds.
\newblock \emph{Boletim da Soc. Bras. de Mat.}, 2:\penalty0 23--36, 1971.

\bibitem[Truesdell(1959)]{TRUESDELL1959}
C.~Truesdell.
\newblock Invariant and complete stress functions for general continua.
\newblock \emph{Arch. Rational Mech. Anal.}, 4:\penalty0 1--27, 1959.

\bibitem[Wolf(2011)]{Wolf2011}
J.~A. Wolf.
\newblock \emph{Spaces of Constant Curvature}.
\newblock American Mathematical Society, Providence, RI, 2011.

\bibitem[Yavari(2013)]{Yavari2013}
A.~Yavari.
\newblock Compatibility equations of nonlinear elasticity for
  non-simply-connected bodies.
\newblock \emph{Arch. Rational Mech. Anal.}, 209:\penalty0 237--253, 2013.

\end{thebibliography}

\end{document}